\documentclass[journal]{IEEEtran}
\usepackage{amsmath,amsfonts,amsthm}
\usepackage{color}
\newcommand{\dn}{\mathbf{d}}
\newcommand{\R}{\mathbb{R}}
\newcommand{\interior}{\mathrm{int}}
 \newcommand{\zero}{\mathbf{0}}

\newtheorem{theorem}{Theorem}
\newtheorem{corollary}{Corollary}
\newtheorem{example}{Example}
\newtheorem{lemma}{Lemma}
\newtheorem{remark}{Remark}
\renewcommand{\top}{\rm T}
\usepackage{mathtools,amssymb}

\begin{document}
\title{Finite- and Fixed-Time Nonovershooting Stabilizers and Safety Filters by Homogeneous Feedback}
\author{Andrey Polyakov and Miroslav Krstic
\thanks{Andrey Polyakov is with Inria, Lille, France, {andrey.polyakov@inria.fr}.

Miroslav Krstic is with the University of California, San Diego, USA, {krstic@ucsd.edu}. 
}}
\maketitle

\begin{abstract}
    Non-overshooting stabilization is a form of safe control where the setpoint chosen by the user is at the boundary of the safe set. Exponential non-overshooting stabilization, including suitable extensions to systems with deterministic and stochastic disturbances, has  been solved by the second author and his coauthors. In this paper we develop homogeneous feedback laws for fixed-time nonovershooting stabilization for nonlinear systems that are input-output linearizable with a full relative degree, i.e., for systems that are diffeomorphically equivalent to the chain of integrators. These homogeneous feedback laws can also assume the secondary role of `fixed-time safety filters' (FxTSf filters) which keep the system within the closed safe set for all time but, in the case where the user's nominal control commands approach to the unsafe set, allow the system to reach the boundary of the safe set no later than a desired time that is independent of nominal control and independent of the value of the state at the time the nominal control begins to be overridden. 
\end{abstract}

\section{Introduction}

Nonovershooting control is aimed at solving  a  tracking problem under certain state/output constrains (see, e.g. \cite{WielandAllgower2007:NOLCOS}, \cite{KrsticBement2006:TAC}). Such a control can be  employed in a ''safety filter'' framework to override  a potentially unsafe nominal controller  \cite{abel2022prescribedtime}.
Nonovershooting and safety controllers have been designed for both linear \cite{PhillipsSeborg1988:IJC}, \cite{Longchamp_etal1993:Aut}, \cite{DarbhaBhattacharrya2003:TAC} and nonlinear plants \cite{KrsticBement2006:TAC}, \cite{LindemannDimarogonas2019:IEEECSL}, \cite{Garg_etal2022:IEEE_CSL}.
Inspired by robust non-overshooting control under deterministic disturbances in~\cite{KrsticBement2006:TAC}, i.e., by stabilization to an equilibrium at the barrier along with input-to-state safety (ISSf), mean-square stabilization of stochastic nonlinear systems to an equilibrium at the barrier, along with a guarantee of non-violation of the barrier in the mean sense, is solved in~\cite{9280364}. 

In many cases,  the corresponding synthesis was based on the so-called control barrier functions \cite{Ames_etal2017:TAC}, \cite{Jankovic2018:Aut}, which have
been extensively used in various control applications such as automotive systems \cite{Ames_etal2015:CDC}, \cite{Rahman_etal2021:ACC} and multi-agent robotics \cite{Wang_etal2017:TR}, \cite{SantilloJankovic2021:ACC}.    
The exponential control barrier functions were reported in \cite{NguyenSreenath2016:ACC} and allowed the use of simple linear tools for a control design.
The nonovershooting prescribed-time control  has been developed recently \cite{abel2022prescribedtime} for the integrator chain based on a novel
time-varying backstepping procedure  \cite{Song_etal2017:Aut,KRISHNAMURTHY2020108860}. Such a control allows safety requirements to be fulfilled for a prescribed time interval $[0,T]$ with $T>0$.  Prescribed-time controllers \cite{Song_etal2017:Aut}, which have been developed even for PDEs \cite{Espitia_etal2019:Aut}, form a special subset of both finite-time \cite{BhatBernstein2000:SIAM_JCO} and fixed-time \cite{Polyakov_etal2015:Aut} control algorithms. The concepts of finite-time and fixed-time barrier functions were introduced in \cite{Li_etal2018:IROS} and \cite{Garg_etal2022:IEEE_CSL}, respectively. 

\subsection{Homogeneous nonovershooting stabilization} 
Inspired by the mentioned results, this paper deals with generalized homogeneous nonovershooting controller design. The goal is to synthesize feedback laws that ensure finite-time stabilization (FnTS) and fixed-time stabilization (FxTS) without the overshoot of the system's output. 

Homogeneity is a dilation symmetry studied in many branches of pure and applied mathematics.
For instance, homogeneous differential equations and homogeneous control systems are well studied in the literature \cite{Zubov1958:IVM}, \cite{Hermes1986:SIAM_JCO}, \cite{Kawski1991:ACDS}, \cite{CoronPraly1991:SCL}, \cite{Grune2000:SIAM_JCO}, \cite{Hong2001:Aut}. Being similar (in some sense) to linear control systems, they may demonstrate faster convergence \cite{BhatBernstein2005:MCSS}, better robustness \cite{Andrieu_etal2008:SIAM_JCO} and smaller "peaking effect" \cite[Chapter 1]{Polyakov2020:Book}. To the best of the authors' knowledge, a homogeneous nonovershooting control has never been designed even for the integrator chain.

Finite-time stabilization (FnTS) without an overshoot is formulated as the problem of constructing a feedback law (for example, of a homogeneous type) such that, given an initial state $x_0\in \R^n$,  the feedback guarantees that the trajectory of the closed-loop system initiated in $x_0$ at the time instant $t=0$ reaches the origin no later than a time instant $T(x_0)$ without overshoot in the first coordinate. If, given arbitrary $T>0$, the origin is reached by the time $T$ for any $x_0$, such a non-overshooting feedback is referred to as fixed-time stabilizing (FxTS).

\subsection{Homogeneous safety filters}
To prevent overshoots that may result from the user's nominal control, we also design what are commonly referred to as `safety filters' and have a particularly simple structure for integrator chains \cite{abel2022prescribedtime}. 

\paragraph{Finite- and fixed-time safety}
Since the notion of fixed-time barrier functions has already appeared in the literature \cite{Garg_etal2022:IEEE_CSL}, let us explain what {\em fixed-time safety} (FxTSf) means to us. While the precise meaning of this property and the weaker {\em finite-time safety} (FnTSf) property will become clear in the statements of Theorems \ref{thm:FxTSf} and \ref{thm:FnTSf}, respectively, here we only give a descriptive, intuitive definition. We consider control systems subject to a nominal control $u_{\rm nom}$ and equipped with a ``safety filter'' feedback law $u = {\cal F}(u_{\rm nom}, x)$. The safety filter is said to ensure FxTSf with fixed time $T$ if it ensures safety {\em and} the following holds for all initial conditions inside the safe set: whenever the applied control ${\cal F}(u_{\rm nom}(t), x(t))$ is different from the nominal control $u_{\rm nom}(t)$ over an entire time interval of duration of $T$, the solution reaches the boundary of the safe set no later than at the end of that time interval.

Clearly, FnTSf and FxTSf 
impart {\em less safety} on the system than the conventional exponential and asymptotic safety properties, under which  the trajectories remain away from the boundary for all finite times. This is precisely the point as FnTSf/FxTSf filters are less conservative. They allow the nominal performance to be achieved sooner, with less distortion. In the language of the ``safety vs. agility'' tradeoff \cite{Ames_etal2017:TAC}, the FxTSf filters are the most agile filters available. 

A comparison with the {\em prescribed-time safety} (PTSf) filters in \cite{abel2022prescribedtime} is in order. Both the PTSf and FxTSf filters ensure safety while allowing convergence to the safety boundary in an amount of time that is independent of the initial condition. The difference is that FxTSf are applicable in infinite-time safety operation, whereas PTSf are ideal when a safety prohibition is only of finite duration. 

\paragraph{Release time and restraint time}
The time $T$ figures in both the PTSf and FnTSf/FxTSf approaches. It is important to understand the different roles that $T$ plays in the two approaches. In PTSf, $T$ is the time after which the prohibition of entering the unsafe set is lifted. Hence, $T$ should be understood in PTSf as {\em release time} (from the safety constraint). In FnTSf/FxTSf, the safety constraint applies in perpetuity---the unsafe set is prohibited for all time. So, in FnTSf/FxTSf, $T$ is the length of a period that can commence at any moment and, if $ u_{\rm nom}$ causes itself to be overridden  by the safety filter ${\cal F}$ over the entire interval, the trajectory will necessarily be permitted to reach the boundary of the unsafe set, but not enter it. Hence in FnTSf/FxTSf, $T$ should be understood as {\em restraint time} since the safety filter restrains the trajectory away from the boundary for no longer than $T$. The release time is a property bestowed upon the system by a PTSf safety filter. Likewise, the restraint time is a property bestowed upon the system by the FnTSf/FxTSf safety filter ${\cal F}$. In FnTSf the restraint time $T$ may depend on the value of $x$ at the moment the safety override kicks in and on the nominal control $u_{\rm nom}$, whereas in FxTSf the restraint time can be arbitrarily assigned with the filter and is independent of $x$ and $ u_{\rm nom}$.

\subsection{Contributions and organization} To summarize, this paper's contributions are the design and the associated stability and safety analyses for
\begin{enumerate}
    \item homogeneous finite/fixed-time nonovershooting stabilizing controllers;
    \item homogeneous finite/fixed-time safety filters. 
\end{enumerate}
The paper is organized as follows. Section \ref{sec:problem} presents the problem statement. Short preliminaries
about the homogeneity are given in Section \ref{sec:hom}. A homogeneous nonovershooting controller design and issues of its tuning are discussed in Section \ref{sec:hom_stab}.  Section \ref{sec:hom_safe} is devoted to homogeneous safety filters. As an example,  linear, homogeneous and prescribed-time safety
filters for the double integrator are compared in Section \ref{sec:ex}. Section \ref{sec:con} presents concluding remarks.  Some auxiliary results are given in Section \ref{sec:app}.


\subsection{Notation}
\begin{itemize}
\item $\R$ is the field of reals and $\R_+=\{\alpha \in \R: \alpha> 0\}$.
\item $e_i=(0,\ldots,0,1,0,\ldots,0)^{\top}\in \R^n$ is the $i$-th element of the canonical Euclidean basis in $\R^n$.  
\item $\Sigma_{-}=\{x\in\R^{n}: e^{\top}_1x\leq 0\}$ is the half-space in $\R^n$ that specifies the safety zone of the system.
\item $\interior S$ denotes an interior of the set $S\subset\R^n$.
\item $P\succ 0$ (resp. $P\prec 0$) means that the symmetric matrix $P\in \R^{n\times n}$ is positive (resp. negative) definite.
\item The matrix $P^{\frac{1}{2}}\!=\!M$ is such that $M^2\!=\!P$. By definition, $P^{-\frac{1}{2}}\!=\!\left(P^{\frac{1}{2}}\right)^{-1}$. Notice, for any matrix $P\succ 0$ there always exists a unique matrix $P^{\frac{1}{2}}\succ 0$.  

\item $\lambda_{\max}(P)$ (resp. $\lambda_{\min}(P)$) denotes a maximal (resp. minimal) eigenvalue of the symmetric matrix $P$.
\item $\mathrm{diag} \{p_1,\ldots,p_n\}\in \R^{n\times n}$ denotes a diagonal matrix with the elements $p_1$, $p_2$,\ldots,$p_n$ on the main diagonal.
\end{itemize}
\section{Problem statement}\label{sec:problem}
Consider the system
\begin{equation}
\dot x=Ax+Bu, \quad t>0, \quad x(0)=x_0,\label{eq:main_system}
\end{equation}	
where $x\in \R^n$, $u\in \R$ and 
\begin{equation}\label{eq:AandB}
A=\left(
\begin{array}{cccccc}
 0 & 1& 0 & \cdots  &0\\
 \vdots & \ddots& \ddots & \ddots &  \vdots\\
  \vdots &   & \ddots &\ddots & 0\\
  \vdots &   & & \ddots  &1\\
   0 & \cdots  & \cdots& \cdots &0
\end{array}
\right), \quad B=\left(
\begin{array}{c}
 0\\
\vdots\\
0\\
1
\end{array}
\right).
\end{equation}

Given $x_0\in \interior \Sigma_-$ and $T>0$, we need to design a state feedback control $u=u_{\rm h}(x)$ such that 
\begin{itemize}
	\item[1)] $u_{\rm h}\in C(\R^{n}\backslash\{\zero\}, \R)$ is a uniformly bounded function of the state;
    \item[2)] the closed-loop system  is globally uniformly finite-time stable\footnote{A system $\dot x=f(x), t>0, x(0)=x_0$ is said to be globally uniformly finite-time stable if it is Lyapunov stable and there exists a locally bounded function $\tilde T:\R^n\rightarrow [0,+\infty)$ such that $x(t,x_0)=\zero, \forall t\geq \tilde T(x_0)$ for any solution $x(t,x_0)$ of the system.} with 
    \begin{equation}
    	 x(t,x_0)=\zero,  \quad \forall t\geq T;
    \end{equation}
    \item[3)]  $x(t,x_0)\in \Sigma_{-}$  for all $t\geq 0$.
\end{itemize}
Such a controller (if it exists) steers the state vector at zero without overshoot in  the first component, i.e., $e_1^{\top}x(t,x_0)\leq 0, \forall t\geq 0$. 
In this paper we design a  \textit{homogeneous nonovershooting finite-time controller} $u_{\rm h}$
allowing also a stabilization in a fixed-time $T>0$ independently of $x_0$.
Our aim also  is to design a `safety filter' \cite{abel2022prescribedtime} which applies the user's nominal control $u_{\rm nom}$ while the system operates in the safe set and overrides $u_{\rm nom}$ with the nonovershooting homogeneous control $u_{\rm h}$ in the domain where the user's nominal control commands operation in the unsafe set.

\section{Preliminaries: Homogeneity}\label{sec:hom}
\subsection{Dilation symmetry of functions and vector fields}
Homogeneity is a symmetry 
of a set or a mapping with respect to a group of transformations called dilation (see, e.g. \cite{Zubov1958:IVM, Hermes1986:SIAM_JCO, Kawski1991:ACDS, BhatBernstein2005:MCSS, Andrieu_etal2008:SIAM_JCO, Polyakov2020:Book}). In this paper we deal only with the following  weighted dilations: 
\begin{equation}\label{eq:dilation}
\dn(s)=e^{sG_{\dn}}, \quad G_{\dn}=\mathrm{diag}\{n-i+1\}_{i=1}^n,
\end{equation}
In a more general case, a linear dilation \cite{Polyakov2020:Book} in $\R^n$ is defined by an arbitrary anti-Hurwitz matrix $G_{\dn}\in \R^{n\times n}$.
Recall \cite{Kawski1991:ACDS} that \textit{a vector field $f:\R^n\rightarrow \R^n$ (a function  $h: \R^n\rightarrow \R$)
is said to be $\dn$-homogeneous of degree $\mu\in \R^n$ if
\begin{equation}
f(\dn(s)x)=e^{\mu s}\dn(s)f(x), \quad\forall x\in \R^n, \quad\forall s\in \R
\end{equation}
\begin{equation}
\mbox{(resp. } h(\dn(s)x)=e^{\mu s}h(x), \quad\forall x\in \R^n, \quad\forall s\in \R\text{)}.
\end{equation}
}
Simple computations show that the linear vector field $x\mapsto Ax$ is $\dn$-homogeneous of degree $-1$: 
\begin{equation}
    A\dn(s)=e^{-s}\dn(s)A, \quad \forall s\in \R
\end{equation}
where the dilation $\dn$ is given by \eqref{eq:dilation} and the matrix $A$ is defined by  \eqref{eq:AandB}.

Homogeneous control system are similar (in some sense)  to linear control systems, but they may have some additional useful properties such as better robustness, faster convergence and smaller overshoot \cite[Chapter 1]{Polyakov2020:Book}.
It is well known \cite{Zubov1964:Book}, \cite{Nakamura_etal2002:SICE}, \cite{BhatBernstein2005:MCSS} that \textit{any asymptotically stable homogeneous system 
\begin{equation}\label{eq:f_system}
    \dot x=f(x),t>0, \quad\quad x(0)=x_0, 
\end{equation} of negative (positive) degree is globally uniformly finite-time
(nearly fixed-time\footnote{The system is globally uniformly nearly fixed-time stable if it is globally uniformly asymptotically stable and for any $r>0$ there exists $T_r>0$ such that $\|x(t,x_0)\|\leq r, \forall t\geq T, \forall x_0\in \R^n$.}) stable}, where $f\in C(\R^n\backslash\{\zero\})$. 
 To guarantee a finite-time stabilization of the system \eqref{eq:main_system},\eqref{eq:AandB}, it is sufficient to design an asymptotically stabilizing homogeneous control $u:\R^n\rightarrow \R$ such that $u(\dn(s)x)=u(x), \forall x\in \R^n, \forall s\in\R$. 
 Several schemes of homogeneous and locally homogeneous control design for the chain of integrators are presented in the literature \cite{CoronPraly1991:SCL, BhatBernstein2005:MCSS, Andrieu_etal2008:SIAM_JCO}. In this paper we utilize the technique proposed in \cite{Polyakov_etal2015:Aut}, since it allows simple rules for control parameters tuning and the settling time adjustment to be applied.

\subsection{Canonical Homogeneous Norm}
Inspired by \cite{Polyakov_etal2018:TAC} let us define the so-called canonical  homogeneous norm  $\|\cdot\|_{\dn}:\R^n\rightarrow [0,+\infty)$ as follows $\|\zero\|_{\dn}=0$ and 
\begin{equation}\label{eq:hom_norm}
\|x\|_{\dn}\!=\!e^{s_x} \;\; : \;\; \|\dn(-s_x)x\|=1, \quad  x\neq \zero, 
\end{equation} 
where $\|x\|=\sqrt{x^{\top} P x}$ is the weighted Euclidean norm,
\begin{equation}
P=P^{\top}\in \R^{n\times n} \;\;:\;\;  PG_{\dn}+G_{\dn}P\succ 0, \;\; P\succ 0.
\end{equation}
The latter matrix inequalities guarantee that the linear dilation $\dn$ is monotone, i.e., the function $s\mapsto \|\dn(s)x\|$ is monotone for any $x\in \R^n$.
One can be shown \cite{Polyakov_etal2018:TAC} that $\|\cdot\|_{\dn}\in C(\R^n)\cap C^{1}(\R^n\backslash\{\zero\})$ and 
\begin{itemize}
    \item $\|\dn(s)x\|_{\dn}=e^{s}\|x\|_{\dn}$ for all $x\in \R^n$ and all $s\in \R$.
    \item $\|x\|=1 \Leftrightarrow \|x\|_{\dn}=1$ and $\|x\|\leq 1 \Leftrightarrow \|x\|_{\dn}\leq 1$.
\end{itemize}
Moreover, for $x\neq \zero$ we have
\begin{equation}\label{eq:d_hom_norm}
    \tfrac{\partial \|x\|_{\dn}}{\partial x}=\|x\|_{\dn}\tfrac{ x^{\top}\dn^{\top}(-\ln \|x\|_{\dn}) P \dn(-\ln \|x\|_{\dn})}{x^{\top}\dn^{\top}(-\ln \|x\|_{\dn}) P G_{\dn}\dn(-\ln \|x\|_{\dn})x}.
\end{equation}
Below we utilize the homogeneous norm $\|x\|_{\dn}$ for both control design and stability analysis (as a Lyapunov function).
Since the homogeneous norm is defined implicitly, then,
to compute  $\|x\|_{\dn}$  for a given $x\neq \zero$, the equation $\|\dn(-s_x)x\|=1$ needs to be solved with respect to $s_x\in \R$. 
 \begin{example}\em
For $G_{\dn}=\left(\begin{smallmatrix}
 2 & 0\\
 0 & 1
\end{smallmatrix}\right)$, by definition, $\|x\|_{\dn}=V$ is a positive definite  solution  
\begin{equation}
\left(
\begin{smallmatrix}
x_1\\
x_2
\end{smallmatrix} \right)^{\top} 
\left(
\begin{smallmatrix}
V^{-2} & 0\\
0 & V^{-1}
\end{smallmatrix} \right)
P\left(
\begin{smallmatrix}
V^{-2} & 0\\
0 & V^{-1}
\end{smallmatrix} \right)
\left(
\begin{smallmatrix}
x_1\\
x_2
\end{smallmatrix} \right)=1, 
\end{equation}
where 
$$
P=\left(
\begin{array}{cc}
p_{11} & p_{12}\\
p_{12} & p_{22}
\end{array} \right)\succ 0. 
$$
For $x\neq 0$ the solution always exists and unique if 
\begin{equation}
P\succ 0, \quad PG_{\dn}+G_{\dn}P\succ 0.
\end{equation}
Hence, the homogeneous norm $\|x\|_{\dn}$ is a unique positive root of the quartic equation
\begin{equation}
V^4-p_{22}x_2^2V^2-2p_{12}x_1x_2V-x_1^2p_{11}=0
\end{equation}
which can be found using the Ferrari formulas (see Appendix). 
\hfill$\diamond$
 \end{example}
 
Therefore, for $n=2$ the canonical homogeneous norm $\|x\|_{\dn}$  can be calculated explicitly. In other cases,  a special numerical algorithm may be utilized in order to compute $\|x\|_{\dn}$ in practice, (see, e.g. \cite{Polyakov_etal2015:Aut}, \cite{Polyakov_etal2019:SIAM_JCO}, \cite{Wang_etal2020:ICRA} for more details).

\section{Homogeneous stabilization without overshoot}\label{sec:hom_stab}
This section presents a two step procedure for a homogeneous nonovershooting control design. First, inspired by \cite{KrsticBement2006:TAC} we  design a linear controller which stabilizes the system exponentially at zero without overshoot. Next, inspired by \cite[Chapter 9]{Polyakov2020:Book} we transform (``upgrade'') a linear controller to a homogeneous one such that an upper estimate of the settling time for the system with the given initial state $x_0\in \interior\Sigma_-$ can be assigned a priori.  
\subsection{Linear nonovershoting control design}
For $\lambda>0, i=1,\ldots,n-1$, let us introduce the row vectors 
\begin{eqnarray}
 h_1&=& - e_1^{\rm T}
\\ 
h_{i}&=&h_{i-1}A+\lambda h_{i-1}, \quad i=2,\ldots,n,
\end{eqnarray}
namely,
\begin{equation}\label{eq:h_i}
    h_i = -e_1^{\rm T}(A+\lambda I)^{i-1}\,,\quad i=1,\ldots,n,
\end{equation}
 and consider the linear cone
\begin{equation}\label{eq:Omega_0}
      \Omega=\{x\in \R^n: h_ix \geq 0, \  i=1,\ldots,n\}\subset\Sigma_{-}.
    \end{equation}
The following lemma is a particular case of the results obtained in \cite{KrsticBement2006:TAC}.  

\begin{lemma}\label{lem:lin_safety}
The linear feedback 
     \begin{equation}\label{eq:linear_control}
        u_{\rm lin}=Kx,  \quad    K=h_{n}(A+\lambda I) = -e_1^{\rm T}(A+\lambda I)^n, 
    \end{equation} stabilizes exponentially the system \eqref{eq:main_system} at $x=0$ and renders the set $\Omega$ 
    strictly positively invariant for the closed-loop system \eqref{eq:main_system}, \eqref{eq:linear_control}.
    Moreover, given initial value $x_0\in \interior \Sigma_-$,  choosing $\lambda>0$ such that
    \begin{equation}\label{eq:lambda} 
    \sum_{j=0}^{i-1} C^{i-1}_{j} e_{i-j}^{\top} x_0\lambda^{j}\!\leq\! 0,
    \;\; C^{i-1}_{j}\!=\!\tfrac{(i-1)!}{(i-j-1)!j!}, \;\; i\!=\!2,\ldots,n,
    \end{equation}
    guarantees that $x_0\in \Omega$. 
    
    \end{lemma}
    
\begin{proof}
Inspired by \cite{KrsticBement2006:TAC} let us consider the  barrier functions $\varphi_i: \R^n \rightarrow \R$ defined as 
\begin{eqnarray}
 \varphi_1&=&-x_1,\\
 \varphi_i&=&\dot \varphi_{i-1} +\lambda \varphi_{i-1}\,,\quad i=2,\ldots,n,
\end{eqnarray}
namely,
\begin{eqnarray}
\varphi_i&=&h_{i-1}\dot x+\lambda h_{i-1}x=h_{i-1}Ax+\lambda h_{i-1}x
=h_ix
\nonumber\\
&:=& 
-e_1^{\rm T}(A+\lambda I)^{i-1} x\,, \quad i=2,\ldots,n,
\end{eqnarray} 
which is compactly represented as
\begin{eqnarray}
    \Phi &=& \left[\begin{array}{c}
    \varphi_1 \\ \vdots \\ \varphi_n
    \end{array}\right] = H x\,,
    \\ 
    H&=& \left[\begin{array}{c}
    h_1 \\ \vdots \\ h_n
    \end{array}\right] = 
    \left[\begin{array}{c}
    -e_1^{\rm T} \\ -e_1^{\rm T} (A+\lambda I) \\ \vdots \\ -e_1^{\rm T} (A+\lambda I)^{n-1}
    \end{array}\right]\,.\label{eq:H}
\end{eqnarray}
The positive orthant $\{\Phi\in \R^n: \varphi_i \geq 0, \  i=1,\ldots,n\}$ is the same as the set $\Omega$. Furthermore,  $H$ is the observability matrix of the pair $(A+\lambda I, -e_1^{\rm T})$. This pair is completely observable by direct verification, using the definition of observability (that the scalar output of this system being identically zero implies the vector state being identically zero). Hence, $H$ is invertible and $x = H^{-1} \Phi$. 
Noting that 
\begin{align}
\dot \varphi_n=&h_{n}\dot x=h_{n-1}(Ax+Bu)=-u_{ \rm lin}+h_{n}Ax
\nonumber \\
=&-\lambda h_{n}x=-\lambda \varphi_n,
\end{align}
the closed-loop system in the coordinates $\varphi_1,\ldots,\varphi_{n-1}$ has the form
\begin{equation}\label{eq:system_phi}
\left\{
\begin{array}{rcl}
\dot \varphi_1&=&-\lambda\varphi_1+\varphi_2,\\
\dot \varphi_2&=&-\lambda\varphi_2+\varphi_3,\\
&\vdots& \\
\dot\varphi_{n}&=&-\lambda\varphi_{n}\,,
\end{array}
\right.
\end{equation}
namely, $\dot \Phi = (-\lambda I + A) \Phi$. 
Since the matrix $-\lambda I + A$ is Metzler, the $\Phi$-system is positive and the set $\{\Phi\in \R^n: \varphi_i \geq 0, \  i=1,\ldots,n\}=\Omega$ is strictly positively invariant.

Let us show next that the inequalities \eqref{eq:lambda} are fulfilled for  some $\lambda>0$. 
Indeed, the inclusion $x_0\in \interior \Sigma_-$ ensures $\varphi_1(x_0)=h_1x_0>0$ and, obviously,
\begin{equation}
    \sum_{j=0}^{i-1} C^{i-1}_{j} e_{i-j}^{\top} x_0\lambda^{j}\!=\!-h_{1} x_0\lambda^{i-1}\!+\!\sum_{j=0}^{i-2} C^{i-1}_{j} e_{i-j}^{\top} x_0\lambda^{j}\!\leq\! 0
\end{equation}
for a sufficiently large $\lambda>0$. 

Let us prove now that 
$\varphi_i(x_0)\geq 0, i=1,2,\ldots,n$. 
On the one hand, for any $i\geq 2$ we derive
\begin{equation}
    \varphi_i(x_0)=h_1(A+\lambda I_n)^{i-1}x_0=h_1 \sum_{j=0}^{i-1} C^{i-1}_{j} A^{i-1-j}\lambda ^{j} x_0.
\end{equation}
On the other hand, for $j=i-1$ we  have $A^{i-1-j}=A^0=I_n$ and $h_1A^{i-1-j}=-e_1^{\top}$; for $j=i-2$ we have $A^{i-1-j}=A$ and $h_1A=-e_2^{\top}$; etc. Hence, we conclude   
\begin{equation}
    \varphi_i(x_0)\!=\!h_1\! \sum_{j=0}^{i-1} C^{i-1}_{j}\! A^{i-1-j}\lambda ^{j} x_0\!=\!-\!\sum_{j=0}^{i-1} C^{i-1}_{j} e_{i-j}^{\top}x_0\lambda ^{j}
\end{equation}
and $\varphi_i(x_0)\geq 0$ provided that the inequalities \eqref{eq:lambda} hold.
This means that $x_0\in \Omega$.

Finally, since $\lambda>0$,  the $\Phi$-system \eqref{eq:system_phi} is exponentially stable at $\Phi=0$. Given that $x(t) =H^{-1} \Phi(t)$ and $\Phi_0 = H x_0$, the  closed-loop system \eqref{eq:main_system}, \eqref{eq:linear_control}, which is governed by $\dot x = \left[A - e_n e_1^{\rm T}(A+\lambda I)^n\right]x$, has the solution $x(t) = H^{-1} {\rm e}^{(-\lambda I + A)t} H x_0$, and this system is exponentially stable at $x=0$. 
\end{proof}

The parameter $\lambda>0$ in Lemma \ref{lem:lin_safety} can be selected using a  simpler but more conservative condition.
\begin{corollary}
If $\lambda>0$ satisfies the inequalities
\begin{equation}\label{eq:lambda_2}
    \lambda \!\geq\! 1-\tfrac{e^{\top}_{k+1}x_0}{e^{\top}_1x_0}C^{i-1}_{i-k-1},  \;\; k\!=\!1,\ldots,i-1, \;\;  i\!=\!2,\ldots,n,
\end{equation}
then the inequalities \eqref{eq:lambda} hold.
\end{corollary}

\begin{proof}
Indeed,
for any $i\geq 2$ we derive
\begin{equation}
\begin{array}{l}
    \sum_{j=0}^{i-1} C^{i-1}_{j} e_{i-j}^{\top} x_0\lambda^{j}
    =C^{i-1}_0 e^{\top}_{i}x_0+
    \nonumber\\
    \lambda(C^{i-1}_{1} e^{\top}_{i-1}x_0\!+\!\lambda(\ldots \lambda(C^{i-1}_{i-2}e^{\top}_2x_0\!+\! \lambda C^{i-1}_{i-1}e^{\top}_1x_0))).
\end{array}
\end{equation}
Since $e_1^{\top}x_0<0$ then the inequalities \eqref{eq:lambda_2} imply that
\begin{equation}\label{eq:lambda_i}
    C^{i-1}_{i-1-k}e^{\top}_{k+1}x_0+\lambda e_1x_0\leq e_1^{\top}x_0, \quad k=1,\ldots,i-1.
\end{equation}
Hence, taking into account  $C^{i-1}_{i-1}=1$ we conclude 
\begin{eqnarray}
& C^{i-1}_{i-2}e^{\top}_2x_0+\lambda C^{i-1}_{i-1}e_1^{\top}x_0\leq e_1^{\top}x_0,
& \\
&
    C^{i-1}_{i-3}e^{\top}_3x_0+\lambda(C^{i-1}_{i-2}e^{\top}_2x_0+ \lambda C^{i-1}_{i-1}e_1^{\top}x_0)
    &\nonumber\\
    &\leq 
    C^{i-1}_{i-3}e^{\top}_3x_0+\lambda e_1^{\top}x_0\leq e_1^{\top}x_0,
& \\
&\vdots& \nonumber\\
&C^{i-1}_0 e^{\top}_ix_0+\lambda(C^{i-1}_{1} e^{\top}_{i-1}x_0+\lambda(...+\lambda(C^{i-1}_{i-2}e^{\top}_2x_0
&\nonumber \\ &
+ \lambda C^{i-1}_{i-1}e^{\top}_1x_0)))\le e_1^{\top}x_0. &
\end{eqnarray}
The proof is complete.
\end{proof}

\subsection{Upgrading linear controller to a homogeneous one}

In this subsection we design a finite-time nonovershooting control by means of an ''upgrade'' (a transformation) of a linear "nonovershooting" feedback to a homogeneous one.

\begin{theorem}\label{thm:hom_safty_control}
Let $K\in \R^{1\times n}$ be the gain of the linear controller \eqref{eq:linear_control} defined by the formula \eqref{eq:linear_control}. Then
\begin{itemize}
\item 
the system of linear matrix inequalities (LMIs)
 \begin{equation}\label{eq:LMI}
		\left\{
		\begin{array}{l}				
		P(A+BK)+(A+BK)^{\top}P\prec 0,\\
		PG_{\dn}+G_{\dn} P\succ 0, \quad P\succ 0, 
		\end{array}
		\right.
		\end{equation} 		
has a solution $P=P^{\top}\in \R^{n\times n}$;
   \item for any given $T>0$ the homogeneous controller  
   
\begin{equation}\label{eq:hom_control}
		u_{\rm h}(x)\!=\!\tilde K\dn\left(-\ln \left\|x/r\right\|_{\dn}\right)x, \quad \tilde K\!=\!K\dn(\tilde s),
		\end{equation}
		\begin{equation}\label{eq:tilde_s}
		    \tilde s\geq \ln \max\left\{\frac{1}{\rho T},1\right\}, \quad r>0
		\end{equation}
		stabilizes the system \eqref{eq:main_system} at $x=\zero$ in a finite time such that 
		\begin{equation}
		     \|x_0\|\leq r \quad \Rightarrow \quad x(t,x_0)=\zero,\quad  \forall t\ge T,
		\end{equation} where $\dn(s)$ is given by \eqref{eq:dilation},
		 \begin{equation}\label{eq:rho}
 \rho\!=\!-\lambda_{\max}\left(Q^{\frac{1}{2}}ZQ^{-\frac{1}{2}}\right)>0,
 \end{equation}
 \begin{equation}\label{eq:Z}
 Z=P(A+BK)+(A+BK)^{\top}P\prec 0,
 \end{equation}
 \begin{equation}
     Q=PG_{\dn}+G_{\dn} P\succ 0,
 \end{equation}
  the homogeneous norm $\|x\|_{\dn}$ is  induced by the weighted Euclidean norm \begin{equation}\label{eq:norm}
      \|x\|=\sqrt{x^{\top}\dn(\tilde s)P\dn(\tilde s)x}, \quad x\in \R^n;
  \end{equation}
 \item the homogeneous cone 
   \begin{equation}\label{eq:Omega}
\Omega_{r}\!=\!\left\{\!x\!\in\! \R^{n}\!: \tilde h_i\dn(-\ln\|x/r\|_{\dn})x\!\geq\! 0, i\!=\!1,...,n\!\right\}
\end{equation}
is a positively invariant set for the closed-loop system \eqref{eq:main_system}, \eqref{eq:hom_control} and 
\begin{equation}
    x_0\in \Omega \quad \Rightarrow \quad x_0\in \Omega_r\subset \Sigma_{-}
\end{equation} provided that $\|x_0/r\|_{\dn}\leq e^{\tilde s}$, where 
$\Omega\subset \Sigma_{-}$ is given by \eqref{eq:Omega_0}, $\tilde h_i=h_i\dn(\tilde s)$ and the row vectors $h_i$  are defined by \eqref{eq:h_i}. 
\end{itemize}
\end{theorem}

\begin{proof}
1)  Let $P=H^{\top}\tilde P H$, where $H$ is defined by \eqref{eq:H} and $\tilde P=\tilde P^{\top}\in \R^{n\times n}$ is a symmetric matrix to be defined below. By Corollary  \ref{cor:q_i} we have $H(A+BK)=(-\lambda I_n+A)H$ and  $HG_{\dn}=\left(G_{\dn}+\lambda(nI_n-G_{\dn})A^{\top}\right) H$.
Since 
\begin{align}
    &P(A+BK)+(A+BK)^{\top}P
    \nonumber \\
    & =H^{\top}\tilde PH(A+BK)+(A+BK)^{\top}H^{\top}\tilde P H
\nonumber \\ &=
    H^{\top}\left(\tilde P(A-\lambda I_n)+(A-\lambda I_n)^{\top}\right)H
\end{align}
and
\begin{align}
    &PG_{\dn}+G_{\dn} P
    =H^{\top}\tilde PHG_{\dn}+G_{\dn}H^{\top}\tilde PH
    \nonumber \\
    & =
H^{\top}\left(\tilde P\left(G_{\dn}+\lambda(nI_n-G_{\dn})A^{\top}\right) \right.
\nonumber \\ & \left. +
\left(G_{\dn}+\lambda(nI_n-G_{\dn})A^{\top}\right)^{\top}\tilde P\right)H    
\end{align}
then, the system of LMIs \eqref{eq:LMI} becomes
\begin{equation}
\left\{
\begin{array}{l}
\tilde P(A-\lambda I_n)+(A-\lambda I_n)^{\top}\tilde P \prec 0\\
\tilde PM +M^{\top}\tilde P\succ 0, \\
\tilde P\succ 0, 
\end{array}
\right.
\end{equation}
where 
\begin{equation}
    M=G_{\dn}+\lambda(nI_n-G_{\dn})A^{\top}.
\end{equation}
If $\tilde P=\mathrm{diag} \{p_1,\ldots,p_n\}$ with $p_i>0$ then $\tilde P\succ 0$ and  the latter system of LMIs becomes
\begin{equation}\label{eq:LMI1_temp}
    \left(
    \begin{smallmatrix}
     2\lambda p_1 & -p_2 & 0 & 0 &...& 0 & 0\\
     -p_2 & 2\lambda p_2 & -p_3 & 0 & ... & 0 & 0\\
     0 & -p_3 & 2\lambda p_3 & -p_4 & ... & 0 & 0\\
     ... & ... & ... & ... & ... & ... & ... \\
     0   & 0 & 0 & 0 & ...& 2\lambda p_{n-1} & -p_n\\ 0   & 0 & 0 & 0 & ...&  -2p_n & \lambda p_n
    \end{smallmatrix}
    \right)\succ 0,
\end{equation}
\begin{equation}\label{eq:LMI2_temp}
   \left(
    \begin{smallmatrix}
     2n p_1  & \lambda p_2 & 0 & 0 &...& 0 & 0\\
     \lambda p_2 & 2(n\!-\!1)p_2  & 2\lambda p_3 & 0 & ... & 0 & 0\\
     0 & 2\lambda p_3 & 2(n\!-\!2) p_3 & 3\lambda p_4 & ... & 0 & 0\\
     ... & ... & ... & ... & ... & ... & ... \\
     0   & 0 & 0 & 0 & ...& 4p_{n-1} & (n\!-\!1)\lambda p_n\\ 0   & 0 & 0 & 0 & ...&  (n\!-\!1)\lambda p_n & 2p_n
    \end{smallmatrix}
    \right)\succ 0.
\end{equation}
Let us fix an arbitrary $p_n>0$, e.g., $p_n=1$. Since $\lambda>0$ then there exists a sufficiently large $p_{n-1}>0$ such that 
\begin{equation}
\left(
    \begin{array}{ccccccc}
     2\lambda p_{n-1} & -p_n\\  -p_n & 2\lambda p_n
    \end{array}
    \right)\succ 0 ,
\end{equation}
\begin{equation}
    \left(
    \begin{array}{ccccccc}
    4p_{n-1} & (n\!-\!1)\lambda p_n\\   (n\!-\!1)\lambda p_n & 2p_n
    \end{array}
    \right)\succ 0.
\end{equation}
simultaneously. Hence, using Schur complement we conclude that there exists a sufficiently large $p_{n-2}>0$ such that
the inequalities 
\begin{equation}
 \left(\!
    \begin{array}{ccccccc}
     2\lambda p_{n-2}&  p_{n-1} & 0\\  
     p_{n-1}& 2\lambda p_{n-1} & -p_n\\  
     0& -p_n & 2\lambda p_n
    \end{array}\!
    \right)\!\succ\! 0, 
\end{equation}    
\begin{equation}
    \left(\!
    \begin{array}{ccccccc}
     6p_{n-2} & \!(n\!-\!2)\lambda p_{n-1}\! & 0\\  
    \!(n\!-\!2)\lambda p_{n-1}\! & 4p_{n-1} & (n\!-\!1)\lambda p_n\\  
    0 & \!(n\!-\!1)\lambda p_n\! & 2p_n
    \end{array}\! \right)\!\succ\!0,
\end{equation}
are fulfilled simultaneously.
Repeating the same consideration subsequently for $p_{n-3}$,\ldots,$p_{1}$ we conclude that the system of LMIs \eqref{eq:LMI} is always feasible.

2) 
Let us show that the homogeneous norm $\|\cdot\|_{\dn}$ is a Lyapunov function of the closed-loop system \eqref{eq:main_system}, \eqref{eq:hom_control}. 
Indeed, using the formula \eqref{eq:d_hom_norm} and  the identities $\dn(s)A=e^{s}\dn(s)A, \dn(s)B=e^{s}B, \forall s\in \R$ we derive
\begin{align}
&\frac{d\|x/r\|_{\dn}}{dt}
\nonumber\\
&= \tfrac{\|x/r\|_{\dn} (x/r)^{\top}\dn(-\ln \|x/r\|_{\dn})\dn(\tilde s)P\dn(\tilde s)\dn(-\ln \|x/r\|_{\dn})(\dot x/r)}{(x/r)^{\top}\dn(-\ln \|x/r\|_{\dn})\dn(\tilde s)P\dn(\tilde s)G_{\dn}
\dn(-\ln \|x/r\|_{\dn})(x/r)}
\nonumber\\
 &=e^{\tilde s}\tfrac{x^{\top}\dn(-\ln \|x/r\|_{\dn})\dn(\tilde s)P(A+BK)\dn(\tilde s) \dn(-\ln \|x/r\|_{\dn})x}{x^{\top}\dn(-\ln \|x/r\|_{\dn})\dn(\tilde s)PG_{\dn}\dn(\tilde s)
\dn(-\ln \|x/r\|_{\dn})x}
\nonumber\\
&\leq -\rho e^{\tilde s}\leq -\max\left\{\frac{1}{T}, \rho \right\}\leq -1/T.\label{eq:dotV}
\end{align}
where the linear matrix inequality \eqref{eq:LMI} and the identity \eqref{eq:rho}
are
utilized on the last step.
The latter means that the closed-loop system \eqref{eq:main_system}, \eqref{eq:hom_control} is globally uniformly finite-time stable and
\begin{equation}
\|x(t,x_0)/r\|_{\dn}\leq \|x_0/r\|_{\dn} -\frac{t}{T}
\end{equation}
as long as $x(t,x_0)\neq \zero$. If $\|x\|\le r$ then $\|x/r\|_{\dn}\le 1$ and $x(t,x_0)=\zero$ for all $t\geq T$.

3) Let us consider the vector-valued homogeneous barrier functions $\phi : \R^n \to \R^n$ defined as follows $\phi(x)=H\dn(\tilde s)\dn(-\ln \|x/r\|_{\dn})x$. We have
\begin{equation}
\begin{split}
 \dot \phi\!=&-H\dn(\tilde s)\frac{\frac{d \|x/r\|_{\dn}}{d t}}{\|x/r\|_{\dn}} G_{\dn}\dn(-\ln \|x/r\|_{\dn})x\\
 &\!+\!H\dn(\tilde s)\dn\!\left(\!-\!\ln \left\|\frac{x}{r}\right\|_{\dn}\right)\!\left(\!Ax\!+\!B\tilde K \dn\!\left(\!-\!\ln \left\|\frac{x}{r}\right\|_{\dn}\right)\!x\!\right)\!.
 \end{split}
\end{equation}
 Denoting $\gamma= -e^{-\tilde s} \frac{d \|x/r\|_{\dn}}{d t}$ and using the identities $\dn(s)A=e^{s}\dn(s)A, \dn(s)B=e^{s}B,  \forall s\in \R$   we derive
 \begin{equation}
      \dot \phi=\frac{e^{\tilde s}}{\|x/r\|_{\dn}} H(A+BK+\gamma G_{\dn})\dn(\tilde s) \dn(-\ln\|x\|_{\dn})x.
 \end{equation}
Using  Corollary \ref{cor:q_i} we obtain
 \begin{equation}\label{eq:dot_phi}
     \dot \phi=\frac{e^{\tilde s}}{\|x/r\|_{\dn}}\Pi\phi,
 \end{equation}
 where
$\Pi\!=-\lambda I_n+A+\gamma G_{\dn}+
     \gamma \lambda (nI_n-G_{\dn})A^{\top}.$
Since $\gamma>0$ and the diagonal matrix $nI_n-G_{\dn}$ is non-negative then 
$\Pi$ is a three diagonal Metzler matrix and the system \eqref{eq:dot_phi} is positive. Therefore, the homogeneous cone $\Omega_r$ is a strictly positively invariant set  of the the closed-loop system \eqref{eq:main_system}, \eqref{eq:hom_control} and, by construction, $\Omega_r\subset \Sigma_-$.

Finally, since $\tilde s\geq 0$ then  
$\|\dn(-\tilde s)x_0/r\|_{\dn}\leq e^{-\tilde s}\|x_0/r\|_{\dn}\leq 1$ provided that $\|x_0/r\|_{\dn}\leq e^{\tilde s}$.
If   
$x_0\in \Omega$ then $h_ix_0\geq 0, i=1,\ldots,n$.
Taking into account the representation 
$
    \phi_i(x_0)=h_i e^{-D_i \ln \|\dn(-\tilde s)x_0\|_{\dn}}x_0
$
 we derive $\phi_i(x_0)\geq 0$ due to Corollary \ref{cor:q_i}, i.e., $x_0\in \Omega_r$.
\end{proof}

The feedback \eqref{eq:hom_control} is discontinuous at zero and smooth away from the origin \cite{Polyakov_etal2015:Aut}, so the solution of the system for the time instances greater than the setting time is defined in the sense of Filippov \cite{Filippov1988:Book}. 

\begin{remark}[Fixed-time nonovershooting control] \em
To guarantee fixed-time convergence of any trajectory of the system \eqref{eq:main_system}, \eqref{eq:hom_control} to zero, it is sufficient to select $r$ in \eqref{eq:hom_control} in a manner dependent on the initial state, for example, as $r=\|x_0\|$ for $x_0\neq \zero$. 
\end{remark}

Theorem \ref{thm:hom_safty_control} presents a scheme for a transformation ("upgrade") of a linear nonovershooting control a homogeneous one. The homogeneous controller guarantees a finite-time stabilization of the system state  at zero without overshoot in the first component.   The main advantage of the  finite-time controller \eqref{eq:hom_control} is the simplicity of the parameters tuning. The parameter $r$ in \eqref{eq:hom_control} defines the maximum norm  of the initial state $x_0$ for which the settling time of the system is bounded by the arbitrary fixed number $T>0$. 
\begin{remark}\em
The controller \eqref{eq:hom_control} is globally uniformly bounded as follows
\begin{align}\label{eq:u_estimate}
|u_{\rm h}(x)|^2=& x^{\top}\dn(-\ln \|x/r\|_{\dn}) \tilde K^{\top} \tilde K    \dn(-\ln \|x/r\|_{\dn})x
\nonumber \\
\leq & r^2\lambda_{\max}\left(P^{-\frac{1}{2}}  K^{\top}  K  P^{-\frac{1}{2}}\right) \text{ for all }x\in \R^n,
\end{align}
where the identity 
$$
x^{\top}\dn(-\ln \|x/r\|_{\dn})\dn(\tilde s)P  \dn(\tilde s)\dn(-\ln \|x/r\|_{\dn})x=r^2
$$ is utilized on the last step. 
\end{remark}

The positive cone\footnote{A set $\mathcal{D}\subset \R^n$ is said to be a positive cone if $x\in \mathcal{D} \Rightarrow e^sx\in \mathcal{D}, \forall s\in \R$.} $\Omega\subset\Sigma_-$ defines a set of initial states of the system \eqref{eq:main_system} for which the linear control \eqref{eq:linear_control} stabilizes the system without overshoot in the first coordinate. In the case of the homogeneous  nonovershooting control \eqref{eq:hom_control} such a positively invariant set is
the $\dn$-homogeneous cone\footnote{A set $\mathcal{D}\subset \R^n$ is said to be a $\dn$-homogeneous cone if $x\in \mathcal{D} \Rightarrow \dn(s)x\in \mathcal{D}, \forall s\in \R$.} $\Omega_r\subset\Sigma_-$. 
Since $|u_{\rm lin}(x)|=|u_{\rm h}(x)|$ for $\|x/r\|_{\dn}=e^{\tilde s}$ then
linear and homogeneous controllers have  the  same maximum magnitude on the homogeneous ball  
\begin{equation}
  B_r=\{x\in \R^n: \|x/r\|_{\dn}\leq e^{\tilde s}\}.\label{eq:hom_ball}
\end{equation}
Theorem \ref{thm:hom_safty_control} implies that  
the positively invariant compact set $\Omega_r \cap\ B_r$ of the homogeneous control system is larger than the positively invariant set $\Omega\cap B_r$ of the linear control system despite that both controllers have the same magnitude on $B_r$.

\begin{remark}\em The simple combination of the linear and homogeneous feedbacks
\begin{equation}\label{eq:lin_hom_control}
 u=\left\{ 
 \begin{array}{ccc}
 u_{\rm h} & \text{ if } & \|x/r\|_{\dn}\leq e^{\tilde s},\\
 Kx & \text{ if } & \|x/r\|_{\dn}> e^{\tilde s}
 \end{array}
 \right.
\end{equation} 
gives a nonovershooting control $u\in C(\R^n\backslash\{\zero\})$ as well. The closed-loop system \eqref{eq:main_system}, \eqref{eq:lin_hom_control} is globally finite-time stable with the   positively invariant set  $\Omega\cup \Omega_r$. Indeed, let us consider the non-smooth  Lyapunov function
\begin{equation}
 V=\left\{ 
 \begin{array}{ccc}
 \|x/r\|_{\dn} & \text{ if } & \|x/r\|_{\dn}\leq e^{\tilde s},\\
 \|\dn(-\tilde s)x/r\| & \text{ if } & \|x/r\|_{\dn}> e^{\tilde s}.
 \end{array}
 \right.
\end{equation}
 By construction, $V$ is continuous on $\R^n$, locally Lipschitz continuous on $\R^n\backslash\{\zero\}$,  continuously differentiable on the set $\R^{n}\backslash\{x\in \R^n:\|x/r\|_{\dn}=e^{\tilde s}\}\backslash\{\zero\}$ and
\begin{equation}
    \dot V\leq\left\{
    \begin{array}{ccc}
-\frac{1}{T} & \text{ if } & 0<V< e^{\tilde s},\\
 -\tilde \rho V & \text{ if } & V> e^{\tilde s},
 \end{array}
 \right.
\end{equation}
where $\tilde \rho=-\lambda_{\max} (P^{-1/2}ZP^{-1/2})>0$. For $V=e^{\tilde s}$ using the chain rule for the Clarke's gradient we derive $\dot V\leq -\min\{1/T, \tilde \rho e^{\tilde s}\}<0$.
The continuous barrier functions $\tilde \phi_i,i=1,\ldots,n$ for this system
\begin{equation}
 \tilde \phi_i=\left\{ 
 \begin{array}{ccc}
\phi_i & \text{ if } & V\leq e^{\tilde s},\\
 \varphi_i & \text{ if } & V> e^{\tilde s}.
 \end{array}
 \right.
\end{equation}
combine smooth linear and homogeneous barrier functions $\varphi$ and $\phi$, respectively.
Since $\dot V<0$ then  the vector $\tilde \phi=(\tilde \phi_1, \ldots,\tilde \phi_n)^{\top}$ satisfies the switched positive system  
\begin{equation}
\frac{d\tilde \phi}{d t} =
\left\{
\begin{array}{ccc}
\frac{e^{\tilde s}}{\|x/r\|_{\dn}} \Pi \tilde \phi & \text{if} & V\leq e^{\tilde s},\\
 (-\lambda I_n \!+\!A) \tilde \phi & \text{if} & V> e^{\tilde s},
 \end{array}
\right.
\end{equation}
having a unique  switching isolated in time.  
\end{remark}

The differential equations \eqref{eq:system_phi} and \eqref{eq:dot_phi} for linear  and homogeneous barrier functions are slightly different. In the linear case, the equation has the constant Metzler matrix $-\lambda I_n +A$, while the tridiagonal Metzler matrix $\Pi$ in the homogeneous case  depends on the time derivative of the Lyapunov function $\|x/r\|_{\dn}$. This dependence complicates the analysis and design of the homogeneous safety filters. 

Despite the obvious difference of positively invariant sets of linear and homogeneous controllers, their mathematical descriptions in polar homogeneous coordinates are identical provided that we do not care about tuning of the settling time of the homogeneous system, namely, if $\tilde s=0$ and $r=1$. Indeed, 
the polar homogeneous coordinates \cite{Praly1997:CDC} in $\R^n$ are given by the polar radius $\rho=\|x\|_{\dn}$  and the unit vector $y=\dn(-\ln \|x\|_{\dn})x$, which defines the so-called homogeneous projection on the unit sphere. In the linear case we deal with the standard dilation $\tilde \dn(s):=e^{I_n s}$, so the polar radius is just the Euclidean norm  $\tilde \rho:=\|x\|_{\tilde \dn}=\|x\|$  and the homogeneous projection  is
$\tilde y:=\tilde \dn(-\ln\|x\|_{\tilde \dn})x=x/\|x\|$. 
Hence, in the polar coordinates we have $\Omega=\{(\tilde \rho,\tilde y) : h_i\tilde y\geq 0,\forall i\}$ and 
$\Omega_r=\{(\rho,y) : h_i y\geq 0, \forall i\}$ provided that $r=1$ and $\tilde s=0$.

The barrier functions as well as the structure of the system \eqref{eq:dot_phi} depend on the particular control law. Some interesting homogeneous control algorithms for the integrator chain are introduced in
 \cite{CoronPraly1991:SCL, BhatBernstein2005:MCSS, Andrieu_etal2008:SIAM_JCO}. However,  a design of barrier functions and a selection of control parameters allowing a nonovershooting finite-time/fixed-time stabilization based on these  algorithms  are unclear yet.

\section{Homogeneous Safety Filters}\label{sec:hom_safe}

A nonovershooting controller can be implemented in practice as the so-called `safety filter'
\cite{abel2022prescribedtime}:
\begin{equation}\label{eq:safety_filter}
    u=\min(u_{\rm nom},u_{\rm lin})
\end{equation}
where $u_{\rm nom}$ is a nominal controller which may have overshoots and $u_{\rm lin}$ is a non-overshooting linear controller whose  gain can be time-invariant as in Lemma \ref{lem:lin_safety} or time-varying as in \cite{abel2022prescribedtime}. Such an implementation  allows overshoots of the nominal controller $u_{\rm nom}$ to be eliminated.  For $n\leq 2$ the same safety filter is applicable in the case of the homogeneous nonovershooting controller \eqref{eq:hom_control}.
For $n\geq 3$ we show that the overshoots of the nominal controller can be eliminated by the homogeneous controller using a modified filter.

Let us define $\Delta_r=+\infty$ for $n\leq 2$ and 
\begin{equation}
\Delta_r(x)= \frac{\gamma_r}{\gamma_u}+\min\limits_{i=2,\ldots,n-1} 
 \tfrac{c_i\phi_i+\phi_{i+1}}{\lambda (i-1)\gamma_u \phi_{i-1}}, 
\end{equation}
\begin{equation}
    \gamma_r=
\tfrac{\phi^{\top}\tilde P (\lambda I_n -A) \phi}
{\phi^{\top}\tilde PHG_{\dn}H^{-1}\phi}>0, \quad \forall x\in \R^n
\end{equation}
\begin{equation}
\gamma_u=
\tfrac{\phi^{\top}\tilde P B}
{\phi^{\top}\tilde PHG_{\dn}H^{-1}\phi}\geq  0, \quad  \forall x\in \Omega_{r}
\end{equation}
where $n\geq 3$, $x\in \interior \,\Omega_r$, $c_i>0$ are arbitrary constants,
$\phi_i=\tilde h_i \dn(-\ln \|x/r\|_{\dn})x$, $i=1,...,n$ are the homogeneous barrier functions
and $\tilde P\succ 0$ is a diagonal matrix defined in the proof of Theorem \ref{thm:hom_safty_control}.
The positive function $\Delta_r$ may have infinite values on the boundary of the set $\Omega_r$:
\begin{equation}
\Delta_r(x):=\liminf_{z\to x} \Delta_r(z)\in (0,+\infty], 
\; z\in \Omega_r, x\in \partial \Omega_r.
\end{equation} 
Finally, let us define $\Delta_r(x)=+\infty$ for $x\in \R^n\backslash \Omega_r$.

\begin{theorem}[Finite-time Safety Filter]\label{thm:FnTSf}
Let  $u_{\rm nom}: [0,+\infty)\mapsto \R$ be a continuous nominal control signal and   $u_{\rm h}:\R^n\mapsto \R$ be a nonovershooting homogeneous feedback defined by Theorem \ref{thm:hom_safty_control} with $P=H^{\top}\tilde P H$, $\tilde P=\mathrm{diag}\{p_1,...,p_n\}, p_n=1$.

Let the safety filter be defined  on $\R^{n}\backslash\{\zero \}$
by the formula
 \begin{equation}\label{eq:safety_filter2}
     u=\max\{u_{\rm h}\!-\!\Delta_r,\, \min\{u_{\rm nom},u_{\rm h}\}\}.
 \end{equation}
 \begin{itemize}
 \item The $\dn$-homogeneous cone $\Omega_r$ is a strictly positively invariant set of 
 the closed-loop system \eqref{eq:main_system}, \eqref{eq:safety_filter2}.
   \item If $x(\tau,x_0)\!=\!\zero$ then 
   \begin{eqnarray}
       u_{\rm nom}(t)\!<\!0,  \forall t\!\in\! (\tau,\tau\!+\!\varepsilon]\;\Rightarrow&-x(t,x_0)\!\in\! \R^{n}_+,  \label{eq:imp_1}\\
       u_{\rm nom}(t)\!\geq\! 0, \forall t\!\in\! [\tau,\tau\!+\!\varepsilon] \;\Rightarrow\!&x(t,x_0)\!=\! \zero,\label{eq:imp_2}
   \end{eqnarray}
   where $\varepsilon>0$ is a positive number. 
    \item For any values of time $\tau>0$ for which
 \begin{equation}   \label{eq-nomh1}
     u_{\rm nom}(t)\geq u_{\rm h}(x(t,x_0)), \;\; \forall t\!\in\! [\tau,\tau+T_0]:x(t,x_0)
     \neq \zero, 
 \end{equation}
  it holds that
 \begin{equation}   \label{eq-nomh2}
     x(\tau+T_0,x_0)=\zero,
 \end{equation}
  where $T_0=\|x(\tau,x_0)/r\|_{\dn}T$.
 \end{itemize}
\end{theorem}
\begin{proof} 1) Repeating the derivation of the formula 
\eqref{eq:dot_phi} in the proof of Theorem \ref{thm:hom_safty_control} we obtain
\begin{equation}\label{eq:phi_safety}
    \dot \phi=e^{\tilde s}\tfrac{(-\lambda I_n +A+\gamma_f HG_{\dn}H^{-1})\phi +B(u_{\rm h}-u)}{\|x/r\|_{\dn}}, 
\end{equation}
as long as $x(t,x_0)\neq \zero$,
where $\gamma_f=-e^{-\tilde s}\frac{d}{dt} \|x/r\|_{\dn}$ and 
$HG_{\dn}H^{-1}\!=\!G_{\dn}\!+\!\lambda (nI_n-G_{\dn})A^{\top}$ by  Corollary \ref{cor:q_i}.

Using the formula \eqref{eq:d_hom_norm} and  the identities $H(A+BK)=(-\lambda I_n+A)H$, $\dn(s)A=e^{s}\dn(s)A, \dn(s)B=e^{s}B, \forall s\in \R$ we derive 
\begin{equation}
\gamma_f=\gamma_r-\gamma_u (u_{\rm h}-u).
\end{equation}
Notice that $\phi^{\top}\tilde PB=\phi_{n}$.
The equation \eqref{eq:phi_safety} can be rewritten in a component-wise form as follows
\begin{equation}
\begin{split}
 \dot \phi_1=&-(\lambda-n\gamma_f)\phi_1+\phi_2,\\
 \dot \phi_i=&-a_i\phi_i+b_i, \quad i=2,\ldots,n-1,\\
 \dot \phi_n=&\lambda \gamma_r(n\!-\!1) \phi_{n-1}\!-\!\left(\lambda \!-\!\gamma_f \!+\!\tfrac{\phi_{n-1}(u_{\rm h} -u)}
{\phi^{\top}\tilde PHG_{\dn}H^{-1}\phi}\right)\phi_n\!+\!u_{\rm h}\!-\!u,\\
 \end{split}
\end{equation}
where $a_i=\lambda -(n-i+1)\gamma_f +c_i $ and
$b_i=\phi_{i+1}+c_i\phi_i+ \lambda (i-1)(\gamma_r-\gamma_u (u_{\rm h}-u))\phi_{i-1}$.

For the the safety filter \eqref{eq:safety_filter2} we have 
\begin{equation}
u_{\rm h}-\Delta \leq u \leq u_{\rm h}, \forall x\in \Omega_r\backslash\{\zero\}
\end{equation}
Hence,  $b_i\geq 0$ as long as $x\in \Omega_r\backslash\{\zero\}$. Consequently, the system \eqref{eq:phi_safety} is positive and any trajectory $x(t,x_0)$ of the closed-loop system 
\eqref{eq:main_system}, \eqref{eq:safety_filter2} with $x_0\in \Omega_r$ belongs to $\Omega_r$ as long as $x(t,x_0)\neq \zero$. 

2) Since $u_{\rm h}$ is discontinuous at $x=\zero\in \Omega_r$ then the origin may be a sliding set of the closed-loop system. In this case, solutions  are defined in sense of Filippov \cite{Filippov1988:Book} and 
\begin{equation}
    \dot x_n\in [u_{\min},u_{\max}] \text{ for } x=\zero, 
\end{equation}
where, due to  $\lim_{z\to \zero} \Delta_r(z)=+\infty$, we have 
\begin{equation}
 u_{\min}= \min\left\{u_{\rm nom},\liminf_{z\to \zero} u_{\rm h}(z)\right\} 
\end{equation}
\begin{equation}
 u_{\max}=  \min\left\{u_{\rm nom},\limsup_{z\to \zero}u_{\rm h}(z)\right\}   
\end{equation}
If $u_{\rm nom}\geq 0$ then $u_{\min}<0\leq u_{\max}$ and the system has a sliding mode at  $x=\zero$ as long as $u_{\max}\geq 0$.
The inequality  $u_{\rm nom}<0$ implies $u_{\min}\leq u_{\max}<0$, so the system leaves the origin, enters the set  $\{x\in \R^n: x_i<0,i=1,...,n\}\subset \interior \, \Omega_r$ and stays there, at least, till $u_{\rm nom}<0$.
 Therefore, the implications \eqref{eq:imp_1}, \eqref{eq:imp_2} hold and $\Omega_r$ is a positively invariant set of the system \eqref{eq:main_system} with the safety filter \eqref{eq:safety_filter2}.

 3) If $u_{\rm nom}(t)\geq u_{\rm h}(x(t,x_0))$ for all $t\in [\tau,\tau+T_0]: x(t,x_0)\neq 
\zero$ then $u(t)=u_{\rm h}(t)$ for all $t\in [\tau,\tau+T_0]: x(t,x_0)\neq 
\zero$. Using  the inequality \eqref{eq:dotV} we derive 
\begin{equation}
\|x(t,x_0)/r\|_{\dn}\leq \|x(\tau,x_0)/r\|_{\dn} -\frac{t-\tau}{T}
\end{equation}
as long as $x(t,x_0)\neq 
\zero$. This means  that there exists $t^*\in [\tau,\tau+T_0]$ such that  
$\|x(t^*,x_0)/r\|_{\dn}=0$. Since $u_{\rm h}(x)>0$ for $x\in \{x\in \R^n: x_i<0,i=1,...,n\}$ then the implications \eqref{eq:imp_1}, \eqref{eq:imp_2}
and the assumption $u_{\rm nom}(t)\geq u_{\rm h}(x(t,x_0))$ for all $t\!\in\! [\tau,\tau+T_0]:x(t,x_0)
     \neq \zero$ guarantee $x(t,x_0)=\zero$ for all $t\in [t^*,\tau+T_0]$.
\end{proof}

\begin{remark}\em 
A safety filter with
the mixed nonovershooting controller \eqref{eq:lin_hom_control}
can also be realized by the formula \eqref{eq:safety_filter2} selecting $\Delta(x)=+\infty$ for $\|x/r\|_{\dn}>e^{\tilde s}$.
\end{remark}  

The key result of this theorem is \eqref{eq-nomh1},    \eqref{eq-nomh2}, which guarantees that the safety boundary is reached by the ``release time'' $T_0$. Note that $\tau$ is not necessarily a single time but it may represent multiple times. Likewise, $\tau + T_0$ may represent multiple times. Depending on the nominal control, the safety filter may kick in and out intermittently and, if the safety override persists, the safety boundary will always be reached in some finite time $T_0$, which is not the same at each of the recurrent occasions of the boundary being reached. 

Theorem \ref{thm:FnTSf} proves that the safety-filter \eqref{eq:safety_filter} keeps the system within the set $\Omega_r\subset \Sigma_-$ for all time but, in the case where the user's nominal control $u_{\rm nom}$ commands operation in the unsafe set, allows the system to reach the boundary of the safe set $\Sigma_-$ no later than a finite-time $T_0>0$. The time $T_0=\|x(\tau,x_0)/r\|_{\dn}T$ is unknown a-priori. It depends implicitly on the user's nominal control $u_{\rm nom}$ and the initial state $x_0\in \Omega_r$. To specify the settling time  independently of $u_{\rm nom}$ and $x_0$, we slightly modify the homogeneous controller \eqref{eq:hom_control} making the parameter $r$ dependent of the system trajectory.

 \begin{theorem}[Fixed-time Safety Filter]\label{thm:FxTSf}
 If 
 \begin{equation}\label{eq:r_t}
  r=r(t):= \max\left\{r_{\min}, \max_{\tau \in [0,t]} \|\dn(-\tilde s) x(\tau)\|\right\}, \quad r_{\min}>0
 \end{equation}
 then under conditions of Theorem \ref{thm:FnTSf} the following holds.
 
 \begin{itemize}
     \item 
     If $x_0\in \Omega$ then 
     $x(t,x_0)\subset \Theta_{r(t)}\subset\Sigma_-$, $\forall t\geq 0$, where 
     \begin{equation}
         \Theta_r=\Omega_r\cap B_r
     \end{equation}
      and $\Omega_r$, $B_r$ are given by \eqref{eq:Omega}, \eqref{eq:hom_ball}, respectively .
      \item If $x(\tau,x_0)\!=\!\zero$ then the implications \eqref{eq:imp_1}, \eqref{eq:imp_2} hold.
    \item 
     For any values of time $\tau\!>\!0$ for which $x(\tau,x_0)\!\in\! \Theta_{r(\tau)}$ and
 \begin{equation}
     u_{\rm nom}(t)\geq u_{\rm h}(x(t,x_0)), \;\; \forall t\!\in\! [\tau,\tau+T]: x(t,x_0)\neq \zero,
 \end{equation}
 it holds that $x(\tau +T,x_0)=\zero,$
 where $T>0$ is a parameter of the homogeneous control \eqref{eq:hom_control}.

 \end{itemize}
 \end{theorem}
 \begin{proof}
  Notice that by definition of the monotone function $r$ we have $\|\dn(-\tilde s)x(t,x_0)/r(t)\|\leq 1$, $\forall t\geq 0$.
  
  1) Since $\Omega_r\subset \Sigma_-$ for any $r\in (0,+\infty)$ (see, Theorem \ref{thm:hom_safty_control}) then $\Theta_r\subset \Sigma_-$ for any $r\in (0,+\infty)$. 

Let us show that $r_1\leq r_2 \Rightarrow \Theta_{r_1}\subset \Theta_{r_2}$.
Indeed, if $x\in \Omega_{r_1}$  and $r_1\leq r_2$ 
then $h_i\dn(s) \dn(-\ln \|x/r_1\|_{\dn}) x=h_i \dn(-\ln \|\dn(-\tilde s)x/r_1\|_{\dn}) x\geq 0$, $i=1,\ldots,n$ and $\|\dn(-\tilde s) x/r_2\|_{\dn}\leq \|\dn(-\tilde s)x/r_1\|_{\dn}\leq 1$. Hence,
using Corollary \ref{cor:q_i} we derive
\begin{equation}
\begin{split}
&h_i \dn(\tilde s) \dn(-\ln \|x/r_2\|_{\dn}) x=\\
&h_i \dn\left(\ln \frac{\|\dn(-\tilde s)x/r_1\|_{\dn}}{\|\dn(-\tilde s)x/r_2\|_{\dn}}\right) \dn(s)\dn(-\ln \|x/r_1\|_{\dn})x\geq 0.
\end{split}
\end{equation}
The system \eqref{eq:main_system}, \eqref{eq:safety_filter2}, \eqref{eq:r_t} is a functional differential equation. The existence of its solution $x(t,x_0)$ for any $x_0\in \R^n$ follows, for example, from \cite[Theorem 4.1]{KolmanovskiMyshkis1992:Book}. Notice also $x(t,x_0)\in \Omega_{r(t)} \Rightarrow \Theta_{r(t)}$ since
$\|x(t,x_0)/r(t)\|\leq e^{\tilde s}$, $\forall t\geq 0$. Moreover,  $x_0\in  \Theta_{r(0)}$ provided that $x_0\in \Omega$.

Let us show that $x(t,x_0)\in\Theta_{r(t)}$ for all $t\geq 0$.
First, we derive the equation \eqref{eq:phi_safety} with
$\gamma=-e^{-\tilde s} \frac{d}{dt} \|x/r\|_{\dn}$ and 
$\phi_i(t_0)\geq0, i=1,
\ldots,n$.
Next, using the formula \eqref{eq:d_hom_norm} and the identities $H(A+BK)=(-\lambda I_n+A)H, \dn(s)A=e^{s}\dn(s)A, \dn(s)B=e^{s}B, \forall s\in \R$ we obtain
\begin{equation}
\gamma=\gamma_r-\gamma_u(u_{\rm h}-u)+\tfrac{\dot r
\phi^{\top}\tilde P\phi}
{r^2\phi^{\top}\tilde PHG_{\dn}H^{-1}\phi}.
\end{equation}
 Since $t\mapsto r(t)$ is a continuous non-decreasing function then 
$r$ is differentiable almost everywhere  and  $\dot r\geq 0$. 

2) Repeating considerations of the proof of Theorem \ref{thm:FnTSf} we conclude $x_0\in \Omega \Rightarrow x(t,x_0)\in \Omega_{r(t)}$ for all $t\geq 0$, and the implications \eqref{eq:imp_1}, \eqref{eq:imp_2} hold.

  3)  If $0\leq t_1<t_2$ are such that
 $
     u_{\rm nom}(t)\geq u_{\rm h}(x(t,x_0))$ and $x(t,x_0)\neq \zero$ for all $\forall t\!\in\! [t_1,t_2] $ then repeating the derivation of  \eqref{eq:dotV} we obtain
  \begin{equation}
     \left.\frac{d}{dt} \|x(t,x_0)/r(t_1)\|_{\dn}\right|_{t=t_1}\leq -\frac{1}{T}
 \end{equation}
 this means that the function $t\to e^{-\tilde s}\|x(t,x_0)/r(t_1)\|_{\dn}=\|\dn(-\tilde s)x(t,x_0)/r(t_1)\|_{\dn}$ is strictly decreasing at $t=t_1$
  and  
 $ r(t)=r(t_1), \forall t\in [t_1,t_2]
 $.
 Indeed, otherwise  there exists $t'\in(t_1,t_2]$ such that $r(t)=r(t_1)$ for all $t\in [t_1,t']$ and $\|\dn(-\tilde s)x(t',x_0)/r(t_1)\|\geq  1$, but this is impossible due to 
 \begin{equation}
     \frac{d}{dt} \|x(t,x_0)/r(t_1)\|_{\dn}\leq -\frac{1}{T}, \;\; \forall t\in[t_1,t']. 
 \end{equation}
 
Hence, repeating the proof of Theorem \ref{thm:FnTSf} we derive $x(\tau+T,x_0)/r(\tau)=\zero$ provided that $u_{\rm nom}(t)\geq u_{\rm h}(x(t,x_0))$ for all $t\!\in\! [\tau,\tau+T]$ such that $x(t,x_0)\neq \zero$.
 \end{proof}
 
 The settling time estimate $T>0$ of the safety filter \eqref{eq:safety_filter2}, \eqref{eq:hom_control}, \eqref{eq:r_t} is independent of the nominal control $u_{\rm nom}$
 and the initial state $x_0\in \R^n$. Such a 'fixed-time safety' is guaranteed by means of an adaptation of the parameter $r$ of the homogeneous controller \eqref{eq:hom_control} (see, the formula \eqref{eq:r_t}). The parameter 
 $r$ depends implicitly on the user's nominal controller $u_{\rm nom}$ and on the initial state $x_0\in \Omega $ of the system. However, this parameter also specifies the maximum magnitude of the homogeneous controller (see, the formula \eqref{eq:u_estimate}) as well as a positively invariant set  of the closed-loop system (see, the formula \eqref{eq:Omega}). In practice, the maximum value of $r$ should be bounded by some $r_{\max}>0$ due to
physical restrictions to admissible control signals. In this case, the  
 safety filtering with the fixed settling time estimate $T>0$ may be  ensured only in the zone $\Theta_{r_{\max}}\subset \Sigma_-$.
 
 If the safety override for a duration of time $T$ recurs, on each occasion the safety boundary will be reached after $T$ time units. The adaptive parameter $r$ is to be different on each of those occasions. 

\section{Example: Double integrator}\label{sec:ex}
\subsection{Linear nonovershooting controller}
Let $n=2$. According to Lemma \ref{lem:lin_safety}, given $x_0\in \interior \Sigma_-\subset \R^2$ a linear nonovershooting  controller can be  defined as follows
\begin{equation}
   u_{\rm lin}(x)=Kx, \quad  K=h_2A+\lambda h_2=(-\lambda^2\;\;\;-2\lambda)
\end{equation}
where
\begin{equation}
   \begin{split}
        h_1=(- 1\;\; 0), \quad h_2=(-\lambda \;\; -1),  \\ \lambda\geq 1-\frac{h_1Ax_0}{h_1x_0}=1-\frac{e_2^{\top}x_0}{e_1^{\top}x_0}. 
   \end{split}
\end{equation}
Indeed, if $\varphi_1=-x_1$ and $
\varphi_2=-\lambda x_1-x_2$, where $x=(x_1,x_2)^{\top}\in \R^2$, then
\begin{equation}\label{eq:dot_varphi_1}
\dot \varphi_1=-\lambda \varphi_1+\varphi_2,
\end{equation}
and
\begin{equation}\label{eq:dot_varphi_2}
    \dot \varphi_2=-\lambda x_2-u_{\rm lin}=-\lambda^2 x_1-\lambda x_2=-\lambda \varphi_2. 
\end{equation}
The system \eqref{eq:dot_varphi_1}, \eqref{eq:dot_varphi_2} is globally asymptotically stable and positive.  Taking into account $\lambda>-\frac{e_2^{\top}x_0}{e_1^{\top}x_0}$
we derive $\phi_1(0)>0$ and $\phi_2(0)>0$, so the controller stabilizes the state vector $x$ at zero without overshoot in the first component: 
\begin{equation}
    e^{\top}_1x(t,x_0)\leq 0, \quad  \forall t\geq 0.
\end{equation}
The set 
\begin{equation}
    \Omega=\{x\in \R^2: h_1x\geq 0, h_2x\geq 0\}
\end{equation} is positively invariant for the closed-loop linear system.
\subsection{Upgrading a linear to a homogeneous controller}
1) \textit{Homogeneous stabilization.} Since for $n=2$ one has $G_{\dn}=\left(\begin{smallmatrix} 2 & 0\\ 0 & 1 \end{smallmatrix}\right)$ then
\begin{align}
  \left(\begin{array}{c} h_1\\h_2\end{array}\right)G_{\dn}=&  \left(\begin{array}{c} h_1\\h_2\end{array}\right)+ \left(\begin{array}{c} h_1\\h_2\end{array}\right)\left(\begin{array}{cc} 1 & 0\\ 0 & 0 \end{array}\right)
\nonumber \\
=&\left(\begin{array}{cc}2 & 0\\ \lambda & 1\end{array}\right) \left(\begin{array}{c} h_1\\h_2\end{array}\right)
\end{align}
and
\begin{align}
    \left(\begin{array}{c} h_1\\h_2\end{array}\right)(A+BK)=&\left(\begin{array}{cc} 0 & -1\\ \lambda^2 & \lambda \end{array}\right)
    \nonumber \\=&\left(\begin{array}{cc} -\lambda & 1\\ 0 & -\lambda \end{array}\right)\left(\begin{array}{c} h_1\\h_2\end{array}\right).
\end{align}
Hence, for  $\alpha>0$ and 
\begin{equation}
P=\left(\begin{array}{c} h_1\\h_2\end{array}\right)^{\top}\left(
\begin{array}{cc}
 \alpha & 0\\
 0 & 1
 \end{array}
\right)\left(\begin{array}{c} h_1\\h_2\end{array}\right),
\end{equation}
the system of LMIs \eqref{eq:LMI} becomes
\begin{align}
    \left(\!\begin{array}{c} h_1\\h_2\end{array}\!\right)^{\top}\!\left(\!
\begin{array}{cc}
 -2\lambda \alpha & 1\\
 1 & -2\lambda
 \end{array}
\!\right)\left(\!\begin{array}{c} h_1\\h_2\end{array}\!\right)\prec & 0,\nonumber\\
    \left(\!\begin{array}{c} h_1\\h_2\end{array}\!\right)^{\top}\!\left(\!
\begin{array}{cc}
 4\alpha & \lambda\\
 \lambda & 2
 \end{array}
\!\right)\left(\!\begin{array}{c} h_1\\h_2\end{array}\!\right)\succ & 0,
\end{align}
or, equivalently,
\begin{equation}
\alpha>\max\left\{\frac{1}{4\lambda^2}, \frac{\lambda^2}{8}\right\}. 
\end{equation}

Selecting $\tilde s=1$, by Theorem \ref{thm:hom_safty_control}, the controller 
\begin{eqnarray}
    u_{\rm h}(x)&=&K\left(\begin{array}{cc} \frac{1}{\|x/r\|_{\dn}^2} & 0\\ 0 & \frac{1}{\|x/r\|_{\dn}} \end{array}\right)x
    \nonumber\\
    &=& - \left(\frac{\lambda}{\|x/r\|_{\dn}}\right)^2 x_1 - 2 \frac{\lambda}{\|x/r\|_{\dn}} x_2
\end{eqnarray}
stabilizes the double integrator to zero in a finite-time and 
\begin{equation}\label{eq:1/rho}
    \|x_0\|\leq r \quad \Rightarrow \quad x(t,x_0)=\zero, \quad t\geq T:=\frac{1}{\rho},
\end{equation} where $r>0$ is an arbitrary parameter,
the homogeneous norm $\|x\|_{\dn}$ is induced by the norm $\|x\|=\sqrt{x^{\top}Px}$ and 
\begin{equation}\label{eq:rho_n=2}
    \rho=-\lambda_{\rm max}\left(\left(\!
\begin{smallmatrix}
 4\alpha & \lambda\\
 \lambda & 2
 \end{smallmatrix}
\!\right)^{-\frac{1}{2}}\left(\!
\begin{smallmatrix}
 -2\lambda \alpha & 1\\
 1 & -2\lambda
 \end{smallmatrix}
\!\right)\left(\!
\begin{smallmatrix}
 4\alpha & \lambda\\
 \lambda & 2
 \end{smallmatrix}
\!\right)^{-\frac{1}{2}}\right)
\end{equation}\
The fixed-time stabilization can be guaranteed by selection $r=\max\{r_{\min}, \|x_0\|\}$, where $r_{\min}>0$.

2) \textit{Homogeneous barrier functions}. The time-derivatives of the functions
\begin{equation}
    \phi_1=-\tfrac{1}{\|x/\gamma\|_{\dn}^2}x_1, \quad \phi_2=-\tfrac{\lambda}{\|x/r\|_{\dn}^2}x_1-\tfrac{1}{\|x/r\|_{\dn}}x_2 
\end{equation}
along the trajectories of the closed-loop homogeneous system 
are
\begin{align}
\dot \phi_1&=\tfrac{2\frac{d}{dt} \|x/r\|_{\dn}}{\|x/r\|_{\dn}^3}x_1-\tfrac{x_2}{\|x/r\|_{\dn}^2}\nonumber
\\
&=-\tfrac{\lambda+2\frac{d}{dt} \|x/r\|_{\dn}}{\|x/r\|_{\dn}}\phi_1+\tfrac{1}{\|x/r\|_{\dn}}\phi_2,\\
\dot \phi_2&=\tfrac{2\lambda \frac{d}{dt} \|x/r\|_{\dn}}{\|x/r\|_{\dn}^3}x_1+\tfrac{-\lambda+\frac{d}{dt}\|x/r\|_{\dn}}{\|x/r\|_{\dn}^2}x_2-\tfrac{1}{\|x/r\|_{\dn}}u
    \nonumber \\
    &=\tfrac{2\lambda \left(-\frac{d}{dt} \|x/r\|_{\dn}\right)}{\|x/r\|_{\dn}}\phi_1\!+\!\tfrac{\lambda^2}{\|x/r\|_{\dn}^3}x_1\!+\!\tfrac{\lambda+\frac{d}{dt} \|x/r\|_{\dn}}{\|x/r\|_{\dn}^2}x_2\nonumber \\
    &=\tfrac{\lambda \left(-\frac{d}{dt} \|x/r\|_{\dn}\right)}{\|x/r\|_{\dn}}\phi_1-\tfrac{\lambda+\frac{d}{dt} \|x/r\|_{\dn}}{\|x/r\|_{\dn}}\phi_2.
\end{align}
Since $\gamma=\left(-\frac{d}{dt} \|x/r\|_{\dn}\right)>0$ and $\|x/r\|_{\dn}>0$ as long as $x\neq \zero$ then the system
\begin{equation}
    \left(
     \begin{smallmatrix}
    \dot \phi_1\\
    \dot \phi_2
    \end{smallmatrix}\right)=
    \tfrac{1}{\|x/r\|_{\dn}}
    \left(
    \begin{smallmatrix}
    -\lambda+2\gamma & 1\\
    \lambda \gamma & -\lambda +\gamma
    \end{smallmatrix}
    \right)
     \left(
     \begin{smallmatrix}
     \phi_1\\
     \phi_2
    \end{smallmatrix}\right)
\end{equation}
is positive and
the set 
 \begin{equation}
 \Omega_r=\left\{x\in \R^2:
 \phi_1(x)\geq 0, \phi_2(x) \geq 0 \right\}.
 \end{equation}
 is a strictly positively invariant set of the closed loop homogeneous system. 
 
3) \textit{Tuning of the settling time}. Let us  optimize the settling time estimate \eqref{eq:1/rho} maximizing $\rho$.
From \eqref{eq:rho_n=2} we conclude that $\rho=-\max\{\eta_1,\eta_2\}$, where
$\eta_1$ and $\eta_2$ are real negative roots of the algebraic equation
\begin{equation}
    \det\left(\left(\!
\begin{array}{cc}
 -2\lambda \alpha & 1\\
 1 & -2\lambda
 \end{array}
\!\right)-\eta \left(\!
\begin{array}{cc}
 4\alpha & \lambda\\
 \lambda & 2
 \end{array}
\!\right)\right)=0,
\end{equation}
which can be rewritten as follows
\begin{equation}
(8\alpha-\lambda^2)\eta^2+2(3+4\alpha)\lambda \eta+4\lambda^2-1=0.  
\end{equation}
Since 
\begin{equation}
D=(3+4\alpha)^2\lambda^2-(8\alpha -\lambda^2)(4\lambda^2-1)>0 \text{ for } \alpha>\tfrac{\lambda^2}{8}
\end{equation}
then 
\begin{equation}
\rho=\frac{(3+4\alpha)\lambda-\sqrt{(3+4\alpha)^2\lambda^2-(8\alpha -\lambda^2)(4\lambda^2-1)}}{(8\alpha-\lambda^2)},
\end{equation}
or, equivalently,
\begin{equation}
    \rho=\frac{4\lambda^2-1}{(3+4\alpha)\lambda+\sqrt{(3+4\alpha)^2\lambda^2-(8\alpha -\lambda^2)(4\lambda^2-1)}}
\end{equation}
Hence, to maximize $\rho$  the parameter $\alpha$ should be minimized. For $\lambda\geq \sqrt[4]{2}$ we have $\alpha>\max\left\{\frac{1}{4\lambda^2},\frac{\lambda^2}{8}\right\}=\frac{\lambda^2}{8}$ and the smallest upper estimate of the settling time $x(t,x_0)=0,\forall t\geq T^*$ corresponds to $\alpha^*=\frac{\lambda^2}{8}$ and
\begin{equation}
T^*=\frac{1}{\rho^*}=\frac{6+\lambda^2}{4\lambda^2-1}\lambda.
\end{equation}
By Theorem \ref{thm:hom_safty_control}, any settling time estimate $T>0$ can be assigned  by means of the modification of the feedback gain 
$\tilde K=K\dn(\tilde s), \tilde s\geq \ln \max\{1/(\rho T),1\}$.
Without such a re-scaling of $K$, the upper estimate of the setting time cannot be less than $T^*$, since the parameter $\alpha>0$ must satisfy the inequality $\alpha>\frac{\lambda^2}{8}$ to guarantee the feasibility of LMIs \eqref{eq:LMI}.

\subsection{Numerical simulations and the use of the homogeneous controller as a ``safety filter''}

For  $x_0=(-4,2)$ we  select 
$$
\lambda=2\geq 1-\frac{e^{\top}_2x_0}{e_1^{\top}x_0}  \quad \text{ and } \quad K=(-4 \;\; -4).
$$
Then
\begin{equation}
h_1=(-1\;\; 0), \quad h_2=(-2\;\; -1), 
\end{equation}
Selecting $\alpha=\frac{\lambda^2+0.01}{8}=0.50125$  we derive $\rho\approx 0.7495$, so the closed-loop system with the homogeneous nonovershooting controller \eqref{eq:hom_control} has the following settling time estimat: $T=\frac{1}{\rho}\approx 1.3342$, i.e., $\|x_0\|\leq r \Rightarrow x(t,x_0)=\zero,\forall t\geq T$.

A nonovershooting controller for the chain of integrators can assume the role of a ``safety filter'' of a nominal controller $u_{\rm nom}$ by being implemented as \cite{abel2022prescribedtime}
\begin{equation}\label{eq-min-filter}
    u=\min(u_{\rm nom},u_{\rm s})
\end{equation}
where the nominal control $u_{\rm nom}$ may demand overshoots, while the override of such an ``unsafe'' nominal control is performed using a ``safe'' control $u_{\rm s}$, which may be either of the linear kind, denoted as $u_{\rm s}=u_{\rm lin}$, or of the homogeneous kind, denoted as $u_{\rm s}=u_{\rm h}$). 
\begin{figure}[t]\centering
\includegraphics[width=90mm]{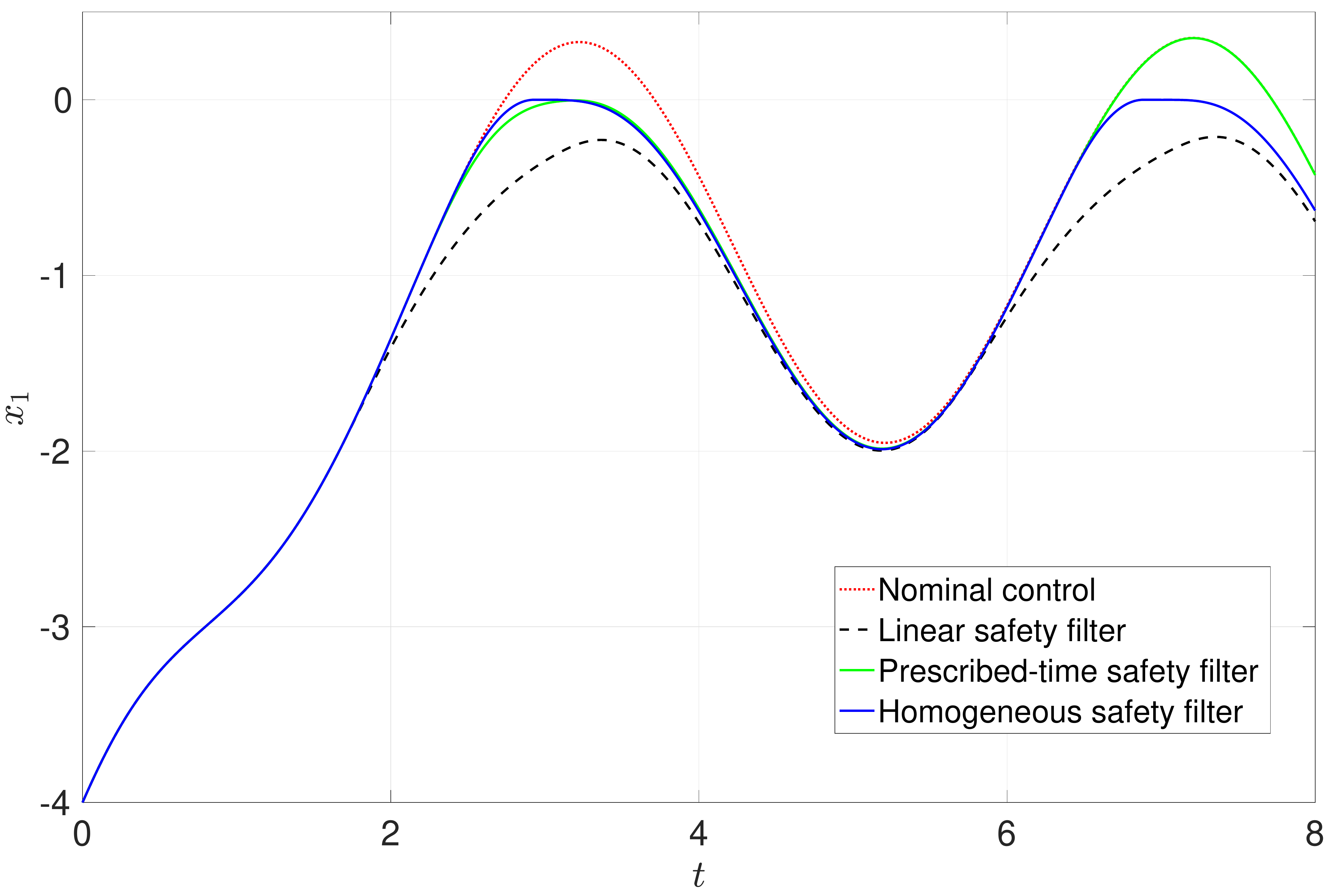}
\caption{Trajectories in $(t,x_1)$-plane for $x_0=(-4\;\; 2)^{\top}$.}
\label{fig:tx_1}
\end{figure}
 \begin{figure}[t]\centering
\includegraphics[width=90mm]{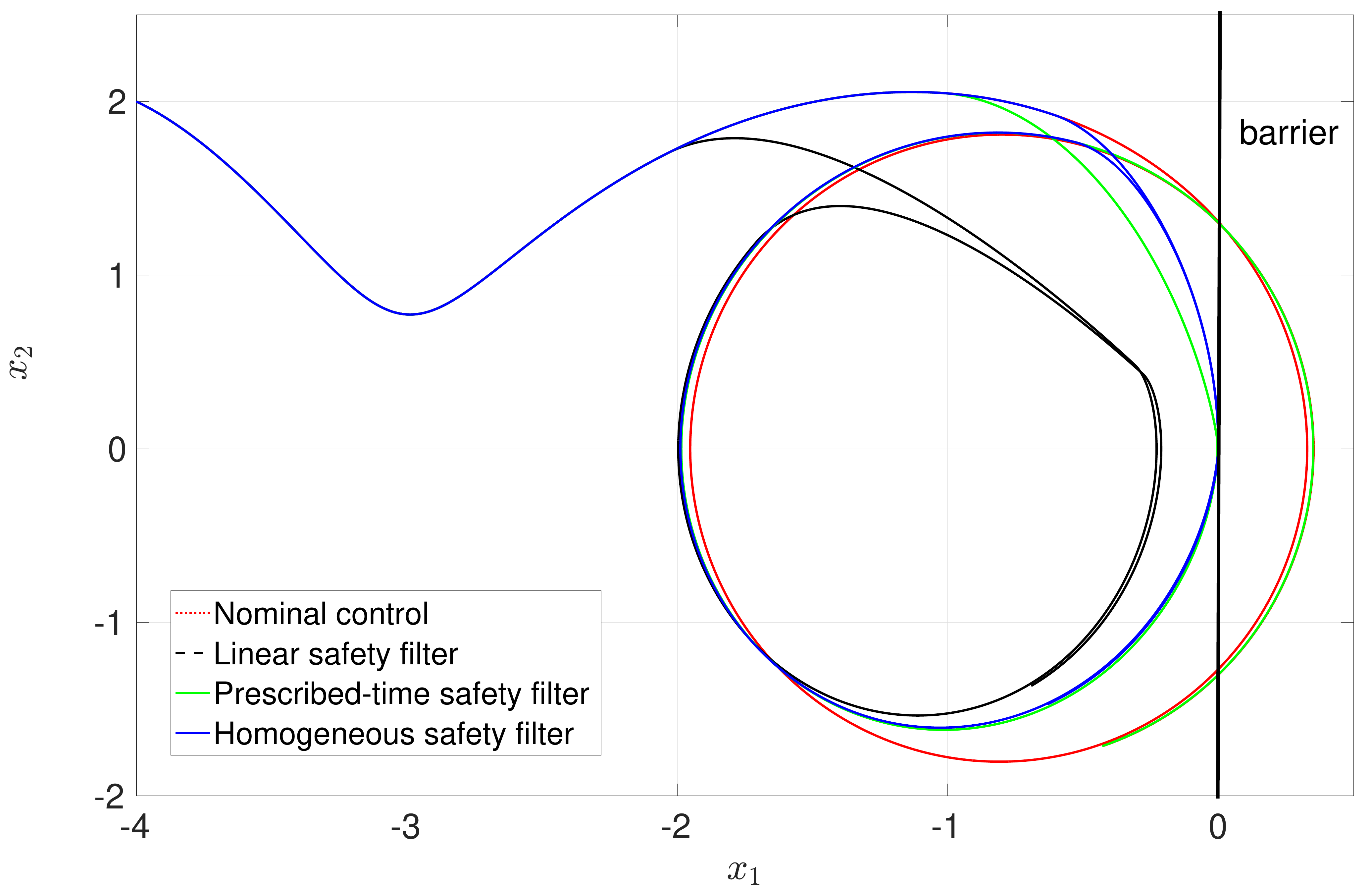}
\caption{Trajectories in $(x_1,x_2)$-plane for $x_0=(-4\;\; 2)^{\top}$.}
\label{fig:x_1x_2}
\end{figure}

To illustrate such ``safety filtering'' through a simulation, we take a nominal control  given by 
\begin{equation}
    u_{\rm nom}=-4\left(x_1+\sin\left(\tfrac{\pi t}{2}\right)+0.8\right)-4\left(x_2+\tfrac{\pi}{2} \cos\left(\tfrac{\pi t}{2}\right)\right)\,,
\end{equation}
which periodically and persistently attempts to violate the safety condition $x_1\leq 0$, while also periodically retreating from such an attempt. In this simulation we compare three safety filters: a linear one based on using  \cite{KrsticBement2006:TAC} with \eqref{eq-min-filter}, a homogeneous one based on the results of this paper, and a prescribed-time (PT) safety filter introduced in  \cite{abel2022prescribedtime}.
The simulation results 
are shown in Figures \ref{fig:tx_1} and \ref{fig:x_1x_2}. 

But let us first examine Figure \ref{fig:Omega}, which presents the positively invariant sets $\Omega$ and $\Omega_r$ for, respectively, the linear nonovershooting controller and homogeneous nonovershooting controller with $r=\|x_0\|=\sqrt{x_0^{\top}Px_0}\approx 6.6348$. 
The sets define the "safety zones" outside of which the corresponding controllers cannot guarantee the absence of overshoots. In the case of the fixed-time safety filter this zone is adaptive and depends of $\max_{\tau \in [0,t]} \|r(\tau)\|$. The larger a zone, the less conservative the override of the nominal controller. Obviously, the homogeneous nonovershooting controller has  a larger positively invariant set than the linear controller (at least close to the origin). The PT controller is a time-varying linear feedback \cite{abel2022prescribedtime}, so its positively invariant set is not defined. 


 \begin{figure}[t]\centering
\includegraphics[width=60mm]{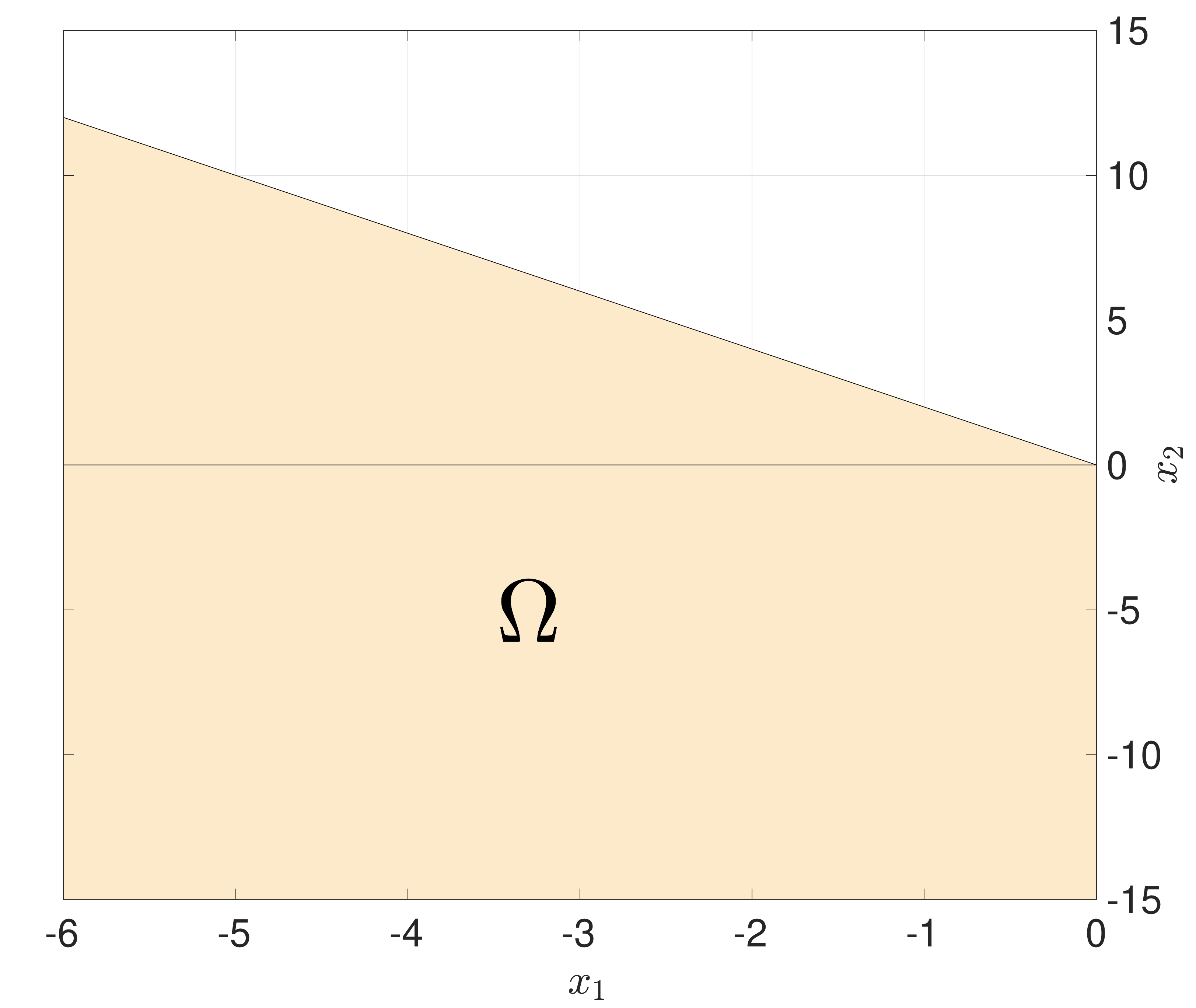}\\
\includegraphics[width=58mm]{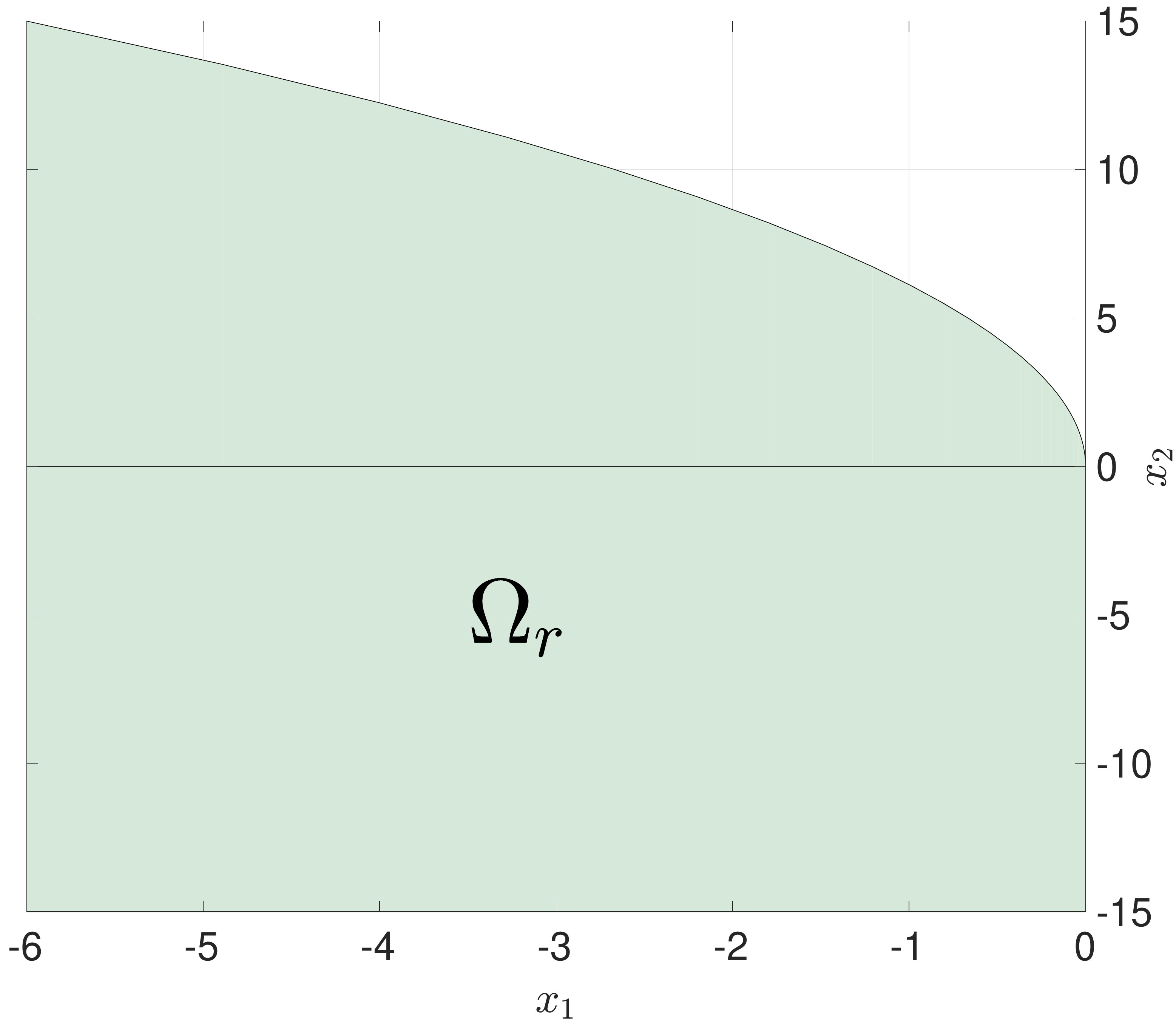}
\caption{Positively invariant sets $\Omega$ and $\Omega_r$  for linear and homogeneous nonovershooting controllers, respectively}
\label{fig:Omega}
\end{figure}

The numerical simulations in Figures \ref{fig:tx_1} and \ref{fig:x_1x_2} show that, in comparison with the linear safety filter, both the homogeneous and the PT safety filters   perform highly unconservative overrides of the nominal controller. The main difference between the PT and homogeneous safety filters is that the operation time of the PT safety filter terminates after the prohibition on $x_1$ being positive ends, which is at $T = 4$ seconds in our case, whereas the homogeneous safety filter continues to restrict the operation to $x_1\leq 0$ for all time. 

However, it is remarkable, and evident from Figure \ref{fig:tx_1}, that the homogeneous safety filter keeps $x_1(t)$ strictly negative for no longer than the fixed time $T=4$ sec (and, in fact, for little more than 1 sec), in response to the attempts of the nominal control to make $x_1$ positive, whereas the more conservative linear safety filter keeps $x_1(t)$ negative all the time, greatly distorting the the resulting (dashed) solution $x_1(t)$ compared to the  nominal solution (red). 

It is evident from Figure \ref{fig:tx_1} that the safety override with the homogeneous safety filter recurs (twice during the time shown; note the two flat tops of the blue curve). This is the result of the fact that the periodic nominal control keeps attempting to violate the safety boundary and keeps retreating. 



%
%

\section{Conclusions}\label{sec:con}

The paper proposes a two-step procedure  for a nonovershooting homogeneous control design for the integrator chain.
First, we construct a linear nonovershooting controller using a backstepping procedure \cite{KrsticBement2006:TAC}. 
Next, we transform (``upgrade'') the linear feedback law to a generalized homogeneous one. The obtained homogeneous  controller is globally uniformly bounded by a number dependent on control parameters $T$ and $r$, which defines a convergence time of the system initiated in the ball of  radius $r$. The main advantage of the proposed scheme is the simplicity of the control design and parameter tuning. Moreover, the  proposed procedure allows a simple upgrade of an existing linear nonovershooting controller to a homogeneous ones. The main disadvantage is the necessity 
to use a special computational procedure (see, e.g. \cite{Polyakov_etal2015:Aut}, \cite{Wang_etal2020:ICRA}) to implement the homogeneous controller in practice for high order systems ($n\geq 3$). 
The positively invariant set of the homogeneous control system is larger than the positively invariant set of the linear control system utilized for the ``upgrade''. This allows a better (less conservative) overriding of a potentially unsafe nominal controller for the safety filter design. 

%
%

\section{Appendix}\label{sec:app}
\subsection{Cubic equation (Cardano formula)}
 Consider the cubic equation
   \begin{equation}
  z^3+pz^2+qz+r=0, \quad p,q,r\in\R.
  \end{equation}
 Introduce the following numbers
 \begin{equation}\label{eq:C}
  C_1=\sqrt[3]{\tfrac{\Delta_1-\sqrt{\Delta_1^2-4\Delta_0^3}}{2}}, \quad   C_2=\sqrt[3]{\tfrac{\Delta_1+\sqrt{\Delta_1^2-4\Delta_0^3}}{2}}
 \end{equation}
 where
 $
  \Delta_0=p^2-3q, \Delta_1=2p^3-9pq+27r.
$

 If $\Delta_1^2-4\Delta_0^3\geq 0$ then
 \begin{equation}
  z_0=-\tfrac{p+C_1+C_2}{3}
 \end{equation}
 is a real root of the cubic equation.
\subsection{Quartic equation (Ferrari formula)}

  Let us consider the quartic equation
  \begin{equation}
   V^4+aV^2+bV+c=0, \quad a,b,c\in\R
  \end{equation}
  and the adjoint cubic equation
  \begin{equation}
  z^3+2az^2+(a^2-4c)z-b^2=0.
  \end{equation}
  If $z_0\in\R$ is a real root of the cubic equation then the roots of the quartic one are 
  \begin{equation}
  \begin{smallmatrix}
   V_1=\frac{-\sqrt{z_0}+\sqrt{-z_0-2a+2b/\sqrt{z_0}}}{2}, &
   V_2=\frac{-\sqrt{z_0}-\sqrt{-z_0-2a+2b/\sqrt{z_0}}}{2},\\
   V_3=\frac{\sqrt{z_0}+\sqrt{-z_0-2a-2b/\sqrt{z_0}}}{2}, &
   V_4=\frac{\sqrt{z_0}-\sqrt{-z_0-2a-2b/\sqrt{z_0}}}{2}.
   \end{smallmatrix}
  \end{equation}
%

\subsection{Auxiliary results}

The following lemma studies some properties of the vectors $h_i$, which are utilized below for a non-overshooting control design.
\begin{lemma}\label{lem:D_i}
If the diagonal matrix $D_{i}\in \R^{n\times n}$ is given by
\begin{align}
 D_i=D_{i-1}+\left(
\begin{smallmatrix}
I_{i-1} & 0\\
0  & 0 
\end{smallmatrix}
\right)=\left( 
\begin{smallmatrix}
 i-1   &    0  &   ... & ... & ... & ...  & 0\\
0   &  i-2 &   ... & ... & ...  & ...  &  0\\
...  &  ...  &  ... & ... &  ...  &  ...  &  ...\\
0   & ...   &   ... & 1  &   0 & ... & 0\\
0  & ...   &  ... &  0 &   0 & ... & 0\\
...  & ...  &   .. & ... &  ... & ... & ...\\
0  & ...  &  ... & 0 &  0 & ... & 0\\
\end{smallmatrix}
\right),& \nonumber\\
i=2,\ldots,n, \quad D_1=0, &
\end{align}
then for $i=2,\ldots,n$ one has
\begin{itemize}
    \item[1)] $D_{i-1}A=AD_i$;
    \item[2)] $Ae^{D_{i}s}=e^{D_{i-1}s}A,\forall s\in \R$;
    \item[3)] $h_i\left(
\begin{smallmatrix}
I_{i} & 0\\
0  & 0 
\end{smallmatrix}
\right)=h_i$ and $h_ie^{s\left(
\begin{smallmatrix}
I_{i} & 0\\
0  & 0 
\end{smallmatrix}
\right)}=e^sh_i,\quad\forall s\in \R$; 
    \item[4)] $h_{i}D_{i}=(i-1)\lambda h_{i-1}$; 
    \item[5)] $h_iD_{n}=(n-i)h_i+(i-1)\lambda h_{i-1}$.
\end{itemize}
\end{lemma}
\begin{proof}
1) Simple calculations show
\begin{align}
   AD_{i}=& \left(
\begin{smallmatrix}
 0 & 1& 0 & ... &0\\
 0 & 0& 1 & ... &0\\
 ... & ...& ... & ... &...\\
  0 & 0& 0 & ... &1\\
   0 & 0& 0 & ... &0
\end{smallmatrix}
\right)\!\left( 
\begin{smallmatrix}
 i-1   &    0  &   ... & ... & ... & ...  & 0\\
0   &  i-2 &   ... & ... & ...  & ...  &  0\\
...  &  ...  &  ... & ... &  ...  &  ...  &  ...\\
0   & ...   &   ... & 1  &   0 & ... & 0\\
0  & ...   &  ... &  0 &   0 & ... & 0\\
...  & ...  &   .. & ... &  ... & ... & ...\\
0  & ...  &  ... & 0 &  0 & ... & 0\\
\end{smallmatrix}
\right)
\nonumber\\ =& 
\left(
\begin{smallmatrix}
 0 & i-2& 0 & ... &0\\
 0 & 0& i-3 & ... &0\\
 ... & ...& ... & ... &...\\
  0 & 0& 0 & ... &0\\
   0 & 0& 0 & ... &0
\end{smallmatrix}
\right)\!=\!D_{i-1}A.
\end{align}
2) Using $D_{i-1}A=AD_i$ we obtain
\begin{align}
    Ae^{D_{i}s}=&\sum_{k=0}^{\infty} \frac{AD^k_{i}s^k}{k!}=\sum_{k=0}^{\infty} \frac{D_{i-1}AD^{k-1}_{i}s^k}{k!}
    \nonumber \\ =&\cdots =\sum_{k=0}^{\infty} \frac{D_{i-1}^kAs^k}{k!}=e^{D_{i-1}s}A.
\end{align}
3) By construction (see, the formula \eqref{eq:h_i}), only first $i$ components of the vector $h_i$ are nonzero. The latter means that $h_{i}\left(
\begin{smallmatrix}
I_{i} & 0\\
0  & 0 
\end{smallmatrix}
\right)=h_i$ and 
\begin{equation}
 h_ie^{s\left(
\begin{smallmatrix}
I_{i} & 0\\
0  & 0 
\end{smallmatrix}
\right)}=\sum_{k=0}^{\infty} \frac{h_i\left(
\begin{smallmatrix}
I_{i} & 0\\
0  & 0 
\end{smallmatrix}
\right)^ks^k}{k!}=\sum_{k=0}^{\infty} \frac{h_is^k}{k!}=e^sh_i.
\end{equation}
4) Since $h_2D_2=(-\lambda, 0,\ldots,0)=(i-1)\lambda h_1$ for $i=2$ then, by induction, for $i\geq 3$ we derive
\begin{align}
    h_{i}D_i=&h_{i-1}AD_i+\lambda h_{i-1} D_{i}
    \nonumber\\
    =&h_{i-1}D_{i-1}A+\lambda h_{i-1} \left(D_{i-1}+\left(
\begin{smallmatrix}
I_{i-1} & 0\\
0  & 0 
\end{smallmatrix}
\right)\right)
\nonumber\\ =&
   (i-2)\lambda h_{i-2}A+(i-2)\lambda^2 h_{i-2}+\lambda h_{i-1}\nonumber\\ =&
(i-2)\lambda h_{i-1}+\lambda h_{i-1}\!=\!(i-1)\lambda h_{i-1}.
\end{align}
5) Since only first $i$ elements of $h_i$ are nonzero then
\begin{align}
h_i D_n=&h_i(I_{n-1}+D_{n-1})=h_i+h_{i}D_{n-1}
\nonumber\\ =&
\cdots =(n-i)h_i+h_iD_i,
\end{align}
where $h_iD_i=(i-1)\lambda h_{i-1}$ as shown above.
\end{proof}
Other useful properties of the vectors $h_i$ and the matrices $D_i$ are given by the following corollary.
\begin{corollary}\label{cor:q_i}
\begin{itemize}
\item[i)] If $\lambda>0$
then 
\begin{equation}
    h_i\dn(s)x\geq  0, \quad \forall s\geq 0
\end{equation}
provided that $h_ix\geq 0$ for all $i=1,\ldots,n$;
\item[ii)] $H(A+BK)=(A-\lambda I_n )H$, where $K\in \R^{1\times n}$ is given by \eqref{eq:linear_control}
and
    \begin{equation}\label{eq:HH}
        H=\left(\begin{smallmatrix} 
        h_1\\
        h_2\\
        \vdots \\
        h_n
        \end{smallmatrix}\right)\in \R^{n\times n};
    \end{equation}
\item[iii)]     \begin{equation}
        HG_{\dn}=\left(G_{\dn}+\lambda(nI_n-G_{\dn})A^{\top}\right) H
    \end{equation}
    where $G_{\dn}=I_n+D_n$.
\end{itemize}
\end{corollary}
\begin{proof}
i). Let us consider the functions $q_i: \R_+ \to \R$ defined as follows
\begin{equation}
    q_i(s)=h_ie^{sD_i}x, \quad s\geq 0. 
\end{equation}
Hence, we have $q_i(0)\geq h_ix$. 
Let us show, by induction, that these functions are non-decreasing and, consequently, non-negative on $\R_+$.
Indeed, since $D_1=0$ then $q_1(s)=h_1e^{sD_1}x=h_1x$ is non-decreasing and non-negative.
For $i\geq 2$  we derive
\begin{align}
    \frac{d}{dt}q_i(s)=
    & h_iD_ie^{sD_i}x=(i-1)\lambda h_{i-1}e^{sD_i}x
    \nonumber\\
    =&(i-1)\lambda e^{s}q_{i-1}(s), 
\end{align}
where the identities $h_iD_i=(i-1)\lambda h_{i-1}$ and $h_{i-1}e^{s\left(
\begin{smallmatrix}
I_{i-1} & 0\\
0  & 0 
\end{smallmatrix}
\right)}=e^sh_{i-1}$ (see, Lemma \ref{lem:D_i}) are utilized.
Since $\lambda>0$ then $\frac{d}{ds}q_{i}(s)\geq 0$ and, consequently, $q_i(s)\geq q_i(0), \forall s\geq 0$ provided that 
$q_{i-1}(s)\geq 0, \forall s\geq 0$.  Hence, taking into account $h_i\dn(s)=e^{(n-i+1)s} h_i e^{D_i s } $ we  conclude that $h_i\dn(s)x\geq 0$ for all $s\geq 0$ provided that $h_ix\geq 0$ for all $i=1,2...,n$.

ii). The identity $H(A+BK)H^{-1}=-\lambda I_n +A$ is proven in 
Lemma \ref{lem:lin_safety} by means of the coordinate transformation $\varphi=Hx$.

iii).  Since $G_{\dn}=D_n+I_n$ then
\begin{align}
HG_{\dn}&=H+\left(\begin{smallmatrix} 
        h_1D_n\\
        h_2D_n\\
        h_3D_n\\
        \vdots\\
        h_nD_n
        \end{smallmatrix}\right)\\
        &= H+
        \left(\begin{smallmatrix} 
        (n-1)h_1\\
        (n-2)h_2 +\lambda h_1\\
        (n-3)h_3+2\lambda h_2\\
        \vdots\\
        (n-1)\lambda h_{n-1}
        \end{smallmatrix}\right)\\
        &=
      H+D_{n}H+
        \lambda\left(\begin{smallmatrix} 
        0\\
        h_1\\
        2h_2\\
        \vdots\\
        (n-1)h_{n-1}
        \end{smallmatrix}\right)\\
        &=
              G_{\dn}H+
        \lambda(nI_n-G_{\dn})\left(\begin{smallmatrix} 
        0\\
        h_1\\
        h_2\\
        \vdots\\
        h_{n-1}
        \end{smallmatrix}\right)\\
        &=
              G_{\dn}H+
        \lambda(nI_n-G_{\dn})A^{\top}H.
\end{align}

\end{proof}

\bibliographystyle{plain}
\bibliography{main.bib}

\end{document}